\documentclass[reqno, eucal]{amsart}

\usepackage[mathscr]{eucal} 
\usepackage[dvips]{graphicx}
\usepackage[dvips]{color}
\usepackage{amsmath,amsfonts,amssymb,amsthm,amscd}
\usepackage{epic}
\usepackage{eepic}
\usepackage{longtable}
\usepackage{array}
\usepackage{here}

\usepackage{geometry}
\newtheorem{theorem}{Theorem}
\newtheorem{lemma}[theorem]{Lemma}
\newtheorem{conjecture}[theorem]{Conjecture}
\newtheorem{proposition}[theorem]{Proposition}

\theoremstyle{definition}

\newtheorem{remark}[theorem]{Remark}

\newcommand{\Z}{{\mathbb Z}}

\newcommand{\C}{{\mathbb C}}

\newcommand{\ok}{{\rm{\bf k}}}

\newcommand{\am}{{\rm{\bf a}}^{\!-} }
\newcommand{\ap}{{\rm{\bf a}}^{\!+} }

\newcommand{\qo}{{\mathcal O}}
\newcommand{\qw}{{\mathcal W}}

\begin{document}

\title[]{New solutions to the tetrahedron equation
\\ associated with 
quantized six-vertex models}

\maketitle

\begin{center}
Atsuo Kuniba,  Shuichiro Matsuike,  Akihito Yoneyama
\vspace{0.3cm}
\\
Institute of Physics, Graduate School of Arts and Sciences,
\\
University of Tokyo, Komaba, Tokyo 153-8902, Japan
\end{center}

\vspace{0.5cm}
\begin{center}{\bf Abstract}
\end{center}

We present a family of new solutions to the tetrahedron equation
of the form $RLLL=LLLR$, where $L$ operator may be regarded as
a quantized six-vertex model whose
Boltzmann weights are specific representations of the 
$q$-oscillator or $q$-Weyl algebras. 
When the three $L$'s are associated with the $q$-oscillator algebra,
$R$ coincides with the known intertwiner of the quantized coordinate ring 
$A_q(sl_3)$.
On the other hand, $L$'s based on the $q$-Weyl algebra lead to  
new $R$'s  whose elements are either factorized or expressed as a terminating 
$q$-hypergeometric type series.

\vspace{0.4cm}

\section{Introduction}

Tetrahedron equation \cite{Zam80} is a key to integrability for lattice models in statistical mechanics in three dimensions.
Among its several versions and formulations,  let us focus on the so-called $RLLL$ relation:
\begin{align}\label{1}
R_{456}L_{236}L_{135}L_{124}  = L_{124}L_{135}L_{236}R_{456}.
\end{align}
Indices here specify the tensor components on which the associated operators act non-trivially.
When the spaces $4, 5, 6$ are evaluated away appropriately, it reduces to the Yang-Baxter equation 
$L_{23}L_{13}L_{12} = L_{12}L_{13}L_{23}$ \cite{Bax}.
Thus (\ref{1}) may be viewed as a quantization of the Yang-Baxter equation along the direction of the 
auxiliary spaces 4, 5 and 6.
It has appeared in several guises and studied from various point of view.
See for example 
\cite{MN89, Kor93, Ser98, BS06, KuS13,Y21} and the references therein.
A survey from a quantum group theoretical perspective is available in \cite{Ku22}.

In this paper we take the spaces 1, 2, 3 as $V=\C^2$ and consider 
the three kinds of $L$ operators:
\begin{align}
L^Z &\in \mathrm{End}(V\otimes V) \otimes \pi_Z(\mathcal{W}_q),\\
L^X &\in \mathrm{End}(V\otimes V) \otimes \pi_X(\mathcal{W}_q),\\
L^O &\in \mathrm{End}(V\otimes V) \otimes \pi_O(\mathcal{O}_q).
\end{align}
They all have the six-vertex model structure \cite{Bax}, i.e., weight conservation property, with respect to the component 
$V \otimes V$.
The last component is taken from specific representations $\pi_X, \pi_Z$ of the 
$q$-Weyl algebra $\mathcal{W}_q$ (\ref{qw}) on $F = \oplus_{m \in \Z} \C|m\rangle$
or $\pi_O$ of the $q$-oscillator algebra $\mathcal{O}_q$ (\ref{qoa}) on 
$F_+ = \oplus_{m \in \Z_{\ge 0}} \C|m\rangle$.
In short, these $L$ operators may be viewed as quantized six-vertex models whose Boltzmann weights are
$\mathrm{End}(F)$ or $\mathrm{End}(F_+)$-valued.
They naturally lead to the generalizations of (\ref{1}) to 
\begin{align}\label{2}
R_{456}L^C_{236}L^B_{135}L^A_{124}  = L^A_{124}L^B_{135}L^C_{236}R_{456},
\end{align}
where A,B and C can be any one of Z, X and O. 
Let us temporarily call it the $RLLL$ relation of type ABC.

The main result of this paper is the explicit solution $R$ for types
ZZZ, OZZ, ZZO, ZOZ, OOZ, ZOO, OZO, OOO, XXZ, ZXX and XZX.
They turn out to be unique up to normalization in each sector specified by a parity condition
in an appropriate sense.
Elements of $R$ are either factorized or expressed as a 
terminating $q$-hypergeometric type series. 
See Table \ref{tab:kekka} in Section \ref{sec:dis} for a summary.
They are new except for type OOO,
where the $RLLL$ relation \cite{BS06} is equivalent 
(cf. Section \ref{ss:OOOint} and \cite[Lem 3.22]{Ku22}) with the 
intertwining relation of the quantized coordinate ring $A_q(sl_3)$, and the $R$ coincides 
with the intertwiner obtained in \cite{KV94}. 
We will show a similar link to $A_q(sl_3)$ also for type ZZZ in Proposition \ref{pr:z}.

The representations $\pi_Z$ and $\pi_X$  of the $q$-Weyl algebra $XZ=qZX$ are natural ones 
in which $Z$ and $X$ become diagonal, respectively. 
See (\ref{piz}) and (\ref{pix}).
They are $q$-analogue of the coordinate and the momentum representations of the canonical commutation relation,
which are formally interchanged via a $q$-difference analogue of the Fourier transformation.
The representation 
$\pi_O$ is a restriction of the special case of $\pi_X$ as explained around (\ref{pio}).
One of our motivation is to investigate systematically how these $L$ operators, 
including their mixtures, 
lead to a variety of solutions $R$ for the associated $RLLL$ relation.
The new $R$'s obtained in this paper 
will be important inputs to many interesting future problems which will be discussed 
in the last section.

The layout of the paper is as follows.
In Section \ref{sec:6v},
the $L$ operators $L^Z, L^X$ associated with the $q$-Weyl algebra and 
$L^O$ for the $q$-oscillator algebra are introduced.
$L^O$ is a restriction of $L^X$, and appeared in the earlier works
\cite{BS06, BMS08, KuS13, KP18, Y21}.
The $RLLL$ relation is formulated.
In Section \ref{sec:ZO} and \ref{sec:ZX}, 
the solutions $R$ are presented 
for the choices $L= L^Z, L^O$ and $L=L^Z, L^X$, respectively.
Some results in the former case can be reproduced as a limit of the latter.
In Section \ref{sec:rep}, a connection to the representation theory of 
$A_q(sl_3)$ is explained. A new result is Proposition \ref{pr:z}. 
Section \ref{sec:dis} contains a summary and discussion on the 
tetrahedron equation of the form $RRRR=RRRR$.
Conjecture \ref{con:r4} is promising.
Appendix \ref{app:eqs} provides the list of explicit forms of the 
$RLLL$ relation for type ZZZ.

\section{Quantized six-vertex models}\label{sec:6v}

We assume that $q$ is generic throughout the paper.

\subsection{$q$-Weyl algebra $\qw_q$ and $q$-oscillator algebra $\qo_q$}

Let $\qw_q$ be the $q$-Weyl algebra,  which is an associative algebra 
with generators $X^{\pm 1}, Z^{\pm 1}$ obeying the relation
\begin{equation}\label{qw}
XZ = qZX
\end{equation}
and those following from the obvious ones 
$X X^{-1} =X^{-1}X  = Z Z^{-1} = Z^{-1}Z = 1$.
Introduce the infinite dimensional vector spaces\footnote{
The actual coefficient field will contain 
many parameters introduced subsequently including $q$.}:
\begin{align}\label{F}
F= \bigoplus_{m \in \Z}\C |m\rangle,
\qquad 
F_+= \bigoplus_{m \in \Z_{\ge 0}}\C |m\rangle.
\end{align}
The algebra $\qw_q$ has irreducible representations 
$\pi_Z$ (resp. $\pi_X$) on $F$ where $Z$ (resp. $X$) is diagonal:
\begin{align}
\pi_{Z}:
\;\;
X|m\rangle = |m-1\rangle, \quad X^{-1}|m\rangle = |m+1\rangle,\quad
Z|m\rangle = q^m|m\rangle, \quad Z^{-1}|m\rangle = q^{-m}|m\rangle,
\label{piz}
\\
\pi_{X}:
\;\;
X|m\rangle = q^m|m\rangle, \quad X^{-1}|m\rangle = q^{-m}|m\rangle,\quad
Z|m\rangle = |m+1\rangle, \quad Z^{-1}|m\rangle = |m-1\rangle.
\label{pix}
\end{align}
They are $q$-analogue of the ``coordinate" and the ``momentum" representations 
of the canonical commutation relation. 

Let $\qo_q$ be the $q$-oscillator algebra, which is an associative algebra 
with generators $\ap, \am, \ok$ obeying the relation
\begin{equation}\label{qoa}
{\rm{\bf k}} \,{\rm{\bf a}}^+ = q\,{\rm{\bf a}}^+{\rm{\bf k}},\quad
{\rm{\bf k}}\,{\rm{\bf a}}^- = q^{-1}{\rm{\bf a}}^-{\rm{\bf k}},\quad
{\rm{\bf a}}^- {\rm{\bf a}}^+ = 1-q^2{\rm{\bf k}}^2,\quad 
{\rm{\bf a}}^+{\rm{\bf a}}^- = 1-{\rm{\bf k}}^2.
\end{equation}
There is an embedding $\iota: \qo_q \hookrightarrow \qw_q$ given by 
\begin{align}\label{emb}
\iota:\;\;
\ok \mapsto X, \quad 
\ap \mapsto Z, \quad
\am \mapsto Z^{-1}(1-X^2).
\end{align}
The composition $\qo_q \overset{\iota}{\hookrightarrow} \qw_q \overset{\pi_X}{\longrightarrow}
\mathrm{End}(F)$ yields the representation:
\begin{align}\label{pio}
{\rm{\bf k}}|m\rangle = q^m |m\rangle,\quad
{\rm{\bf a}}^+|m\rangle = |m+1\rangle,\quad
{\rm{\bf a}}^-|m\rangle = (1-q^{2m})|m-1\rangle.
\end{align}
Due to ${\rm{\bf a}}^-|0\rangle = 0$, the subspace $F_+ \subset F$ becomes invariant and irreducible.
We let $\pi_O: \qo_q \rightarrow \mathrm{End}(F_+)$ denote the resulting irreducible representation 
obtained by restricting (\ref{pio}) to $m \ge 0$.

\subsection{3D $L$ operator}
Let $V = \C v_0 \oplus \C v_1$ be the two dimensional vector space.
We consider $q$-Weyl algebra-valued $L$ operator
\begin{align}
&\mathcal{L} = \mathcal{L}_{r,s,t,w} 
= \sum_{a,b,i,j=0,1} E_{ai}\otimes E_{bj} \otimes \mathcal{L}^{ab}_{ij} 
\in \mathrm{End}(V \otimes V)  \otimes \qw_q,
\label{L1}
\\
&\mathcal{L}^{ab}_{ij}=0\; \text{unless}\; a+b=i+j,
\label{L2}
\\
&\mathcal{L}^{00}_{00} = r,\;\;  \mathcal{L}^{11}_{11} = s,\;\; \mathcal{L}^{10}_{10} = tw X,\;\;
\mathcal{L}^{01}_{01} = -qtX, 
\;\;  \mathcal{L}^{10}_{01} = Z, \;\; \mathcal{L}^{01}_{10} = Z^{-1}(rs- t^2wX^2).
\label{L3}
\end{align} 
Here $r,s,t,w$ are parameters whose dependence has been suppressed in the notation 
$\mathcal{L}^{ab}_{ij}$.
They are assumed to be generic throughout.
The symbol $E_{ij}$ denotes the matrix unit on $V$ acting on the basis as $E_{ij}v_k = \delta_{jk}v_i$.
The $L$ operator $\mathcal{L}$ may be viewed as a quantized six-vertex model
where the Boltzmann weights are $\qw_q$-valued.
See Figure \ref{fig:6vW} for a graphical representation.

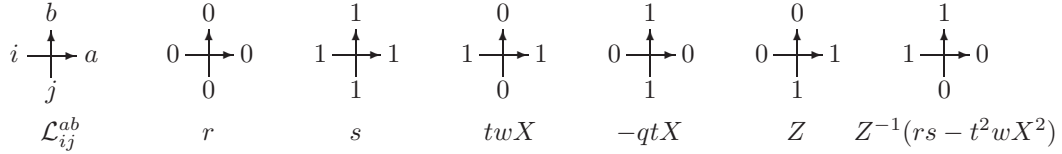
\begin{figure}[H]
\begin{picture}(330,60)(10,30)
{\unitlength 0.011in
\put(6,80){
\put(-11,0){\vector(1,0){23}}\put(0,-10){\vector(0,1){22}}
}
\multiput(81,80.5)(70,0){6}{
\put(-11,0){\vector(1,0){23}}\put(0,-10){\vector(0,1){22}}
}
\put(-74,0){
\put(60.5,77){$i$}\put(77.5,60){$j$}
\put(96,77){$a$}\put(77.5,96.5){$b$}
}
\put(61,77){0}\put(78,60){0}\put(96,77){0}\put(78,96.5){0}
\put(70,0){
\put(61,77){1}\put(78,60){1}\put(96,77){1}\put(78,96.5){1}
}
\put(140,0){
\put(61,77){1}\put(78,60){0}\put(96,77){1}\put(78,96.5){0}
}
\put(210,0){
\put(61,77){0}\put(78,60){1}\put(96,77){0}\put(78,96.5){1}
}
\put(280,0){
\put(61,77){0}\put(78,60){1}\put(96,77){1}\put(78,96.5){0}
}
\put(350,0){
\put(61,77){1}\put(78,60){0}\put(96,77){0}\put(78,96.5){1}
}
\put(78,40){
\put(-77,0){$\mathcal{L}^{ab}_{ij}$}
\put(0,0){$r$} \put(70,0){$s$} \put(134,0){$twX$}
\put(197,0){$-qtX$} \put(278,0){$Z$} \put(310,0){$Z^{-1}(rs-t^2w X^2)$}
}}
\end{picture}
\caption{$\mathcal{L} = \mathcal{L}_{r,s,t,w}$ 
as a $\qw_q$-valued six-vertex model.
Assigning another perpendicular arrow corresponding to 
the $\mathcal{W}_q$-modules leads to a unit of the three dimensional (3D) lattice.
In this context, $L$ will also be called the 3D $L$ operator.}
\label{fig:6vW}
\end{figure}

Note that $\mathcal{L}$ does not contain $X^{-1}$, 
which will be a key in Remark \ref{re:w} below. 
Although $t$ can be absorbed into the normalization of $X$,
we keep it for convenience.
It is easy to see
\begin{align}\label{lwinv}
(\mathcal{L}_{r,s,t,w})^{-1} = (rs)^{-1} \mathcal{L}_{s,r,tw,w^{-1}}.
\end{align}
For the special choice of the parameters $(r,s,t,w) = (1,1,\mu^{-1},\mu^2)$,
$\mathcal{L}$ 
only contains the combinations appearing in the RHS of (\ref{emb}) 
which can be pulled back to the $q$-oscillator algebra.
Therefore we regard it as $\qo_q$-valued, i.e., 
\begin{align}\label{Lmu}
\mathcal{L}_{1,1,\mu^{-1},\mu^2} \in \mathrm{End}(V \otimes V) \otimes \qo_q.
\end{align}
Its elements are given by 
\begin{align}
&\mathcal{L}^{00}_{00} = 1,\;\, \mathcal{L}^{11}_{11} = 1,\;\,
\mathcal{L}^{10}_{10} = \mu \ok, \;\, \mathcal{L}^{01}_{01} = -q\mu^{-1}\ok, \;\, 
\;\, \mathcal{L}^{10}_{01} = \ap, \;\, \mathcal{L}^{01}_{10} = \am.
\label{cLO}
\end{align}
See Figure \ref{fig:6vO}.

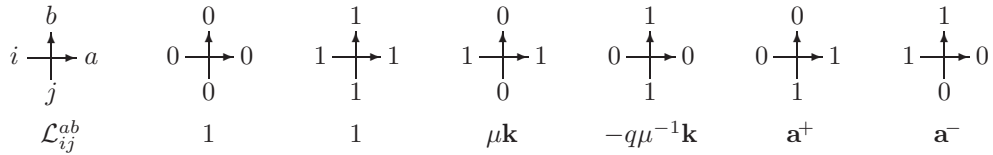
\begin{figure}[H]
\begin{picture}(330,60)(10,30)
{\unitlength 0.011in
\put(6,80){
\put(-11,0){\vector(1,0){23}}\put(0,-10){\vector(0,1){22}}
}
\multiput(81,80.5)(70,0){6}{
\put(-11,0){\vector(1,0){23}}\put(0,-10){\vector(0,1){22}}
}
\put(-74,0){
\put(60.5,77){$i$}\put(77.5,60){$j$}
\put(96,77){$a$}\put(77.5,96.5){$b$}
}
\put(61,77){0}\put(78,60){0}\put(96,77){0}\put(78,96.5){0}
\put(70,0){
\put(61,77){1}\put(78,60){1}\put(96,77){1}\put(78,96.5){1}
}
\put(140,0){
\put(61,77){1}\put(78,60){0}\put(96,77){1}\put(78,96.5){0}
}
\put(210,0){
\put(61,77){0}\put(78,60){1}\put(96,77){0}\put(78,96.5){1}
}
\put(280,0){
\put(61,77){0}\put(78,60){1}\put(96,77){1}\put(78,96.5){0}
}
\put(350,0){
\put(61,77){1}\put(78,60){0}\put(96,77){0}\put(78,96.5){1}
}
\put(78,40){
\put(-77,0){$\mathcal{L}^{ab}_{ij}$}
\put(0,0){1} \put(70,0){1} \put(135,0){$\mu \ok$}
\put(191,0){$-q\mu^{-1}\ok$} \put(278,0){$\ap$} \put(348,0){$\am$}
}}
\end{picture}
\caption{$\mathcal{L} = \mathcal{L}_{1,1,\mu^{-1},\mu^2}$ 
as an $\qo_q$-valued six-vertex model. 
The last two relations in (\ref{qoa}) is a quantization of the 
free Fermion condition 
\cite[Fig. 10.1, eq.(10.16.4)$|_{\omega_7=\omega_8=0}$]{Bax}.}
\label{fig:6vO}
\end{figure}

Now we introduce the three types of (represented) $L$ operators:
\begin{align}
L^Z &= L^Z_{r,s,t,w} = (1 \otimes 1 \otimes \pi_Z)(\mathcal{L}_{r,s,t,w})
\in \mathrm{End}(V \otimes V \otimes F),
\label{LW}
\\
L^X &= L^X_{r,s,t,w} = (1 \otimes 1 \otimes \pi_X)(\mathcal{L}_{r,s,t,w})
\in \mathrm{End}(V \otimes V \otimes F),
\label{LX}
\\
L^O &= L^O_{\mu} = (1 \otimes 1 \otimes \pi_O)(\mathcal{L}_{1,1,\mu^{-1},\mu^2})
\in \mathrm{End}(V \otimes V \otimes F_+).
\label{LO}
\end{align}
From (\ref{lwinv}) and (\ref{Lmu}) we have
\begin{align}\label{Linv}
(L^Z_{r,s,t,w})^{-1} = (rs)^{-1}L^Z_{s,r,tw,w^{-1}},
\quad
(L^X_{r,s,t,w})^{-1} = (rs)^{-1}L^X_{s,r,tw,w^{-1}},
\quad 
(L^O_\mu)^{-1} = L^O_{\mu^{-1}}.
\end{align}

\begin{remark}\label{re:w}
The operator $L^Z$ in (\ref{LW}) keeps the subspace 
$V \otimes V \otimes \bigoplus_{m\le n} \C  |m\rangle \subset F$  invariant 
for any $n \in \Z$.
\end{remark}

\subsection{$RLLL$ relation}
Quantized six-vertex model satisfies the quantized Yang-Baxter equation.
It is a version of the tetrahedron equation having the form of  
the Yang-Baxter equation up to conjugation:
\begin{align}\label{rlll}
R_{456} L_{236}L_{135}L_{124} = L_{124}L_{135}L_{236}R_{456}.
\end{align}
We also call it $RLLL$ relation.
The indices denote the tensor components on which the respective operators act non-trivially.
The operator $L$ will be taken as $L^Z, L^X$ or $L^O$ in (\ref{LW})--(\ref{LO}).
The conjugation operator $R$, which we call 3D $R$ in this paper, 
will be the main object of study in what follows.
In terms of the components of $L$, (\ref{rlll}) reads as 
\begin{align}\label{qybe}
R  \sum_{\alpha, \beta, \gamma}(\mathcal{L}^{\alpha \beta}_{ij} 
\otimes \mathcal{L}^{a \gamma}_{\alpha k} \otimes \mathcal{L}^{bc}_{\beta\gamma})  
= \sum_{\alpha, \beta, \gamma}
(\mathcal{L}^{ab}_{\alpha\beta} \otimes 
\mathcal{L}^{\alpha c}_{i\gamma} \otimes \mathcal{L}^{\beta\gamma}_{jk})\, R
\end{align}
for arbitrary $a,b,c,i,j,k \in \{0,1\}$.
See Figure \ref{fig:qybe}.

\begin{figure}[H]
\begin{picture}(250,60)(-80,-8)

\put(120,0){
\put(0,0){
\put(-27,-8){\line(0,1){58}}\put(-27,-8){\line(1,0){4}}\put(-27,50){\line(1,0){4}}}
\put(-83,18){$= \;\;\; \sum_{\alpha, \beta, \gamma}$}
\put(1,1){\vector(0,1){50}} \put(-1,53.5){$c$}
\put(-13,5){\vector(2,1){55}}\put(45,31){$b$}
\put(-13,40){\vector(2,-1){55}}\put(45,10){$a$}
\put(-18,38){$i$}\put(-20,2){$j$}\put(-1,-8.5){$k$}
\put(11,31){$\alpha$}\put(10,9){$\beta$}\put(-7,22){$\gamma$}
\put(1,0){\put(55,-8){\line(0,1){58}}\put(55,50){\line(-1,0){4}}\put(55,-8){\line(-1,0){4}}
\put(55,18){\; $\circ \; R $}}
}
\put(-190,0){
\put(72,18){$\sum_{\alpha, \beta, \gamma}\; R \; \circ\;$}
\put(156,0){
\put(-27,-8){\line(0,1){58}}\put(-27,-8){\line(1,0){4}}\put(-27,50){\line(1,0){4}}
\put(25,1){\vector(0,1){50}}\put(22.5,53.5){$c$}
\put(-13,31){\vector(2,-1){55}}\put(45,39){$b$}
\put(-13,13){\vector(2,1){55}}\put(45,0){$a$}
\put(-18,31){$i$}\put(-20,10){$j$}\put(23,-8.5){$k$}
\put(11,31){$\beta$}\put(10,10){$\alpha$}\put(28,22){$\gamma$}
\put(1,0){\put(55,-8){\line(0,1){58}}\put(55,50){\line(-1,0){4}}\put(55,-8){\line(-1,0){4}}}
}
}
\end{picture}
\caption{A pictorial representation of 
the quantized Yang-Baxter equation (\ref{qybe}). }
\label{fig:qybe}
\end{figure}

From the conservation condition (\ref{L2}), the equation (\ref{qybe}) becomes $0=0$
unless $a+b+c=i+j+k$.
There are 20 choices of $(a,b,c,i,j,k) \in \{0,1\}^6$ satisfying it.
Among them, the cases $(0,0,0,0,0,0)$ and $(1,1,1,1,1,1)$ 
yield the trivial relation $R(1 \otimes 1 \otimes 1) = (1 \otimes 1 \otimes 1)R$
for any choice of $L = L^Z, L^X, L^O$.
Thus there are 18 non-trivial equations on $R$.
By setting\footnote{$a,b,c,i,j,k$ here are labels of the basis of $F$ or $F_+$ and 
have different meaning from those in (\ref{qybe}) labeling the basis of $V$.}
\begin{align}\label{rdef}
R(|i\rangle \otimes |j\rangle \otimes |k\rangle) = \sum_{a,b,c} R^{a,b,c}_{i,j,k}
 |a\rangle \otimes |b\rangle \otimes |c\rangle,
 \end{align}
they are translated into linear recursion relations on the matrix elements $R^{a,b,c}_{i,j,k}$.
We say that 
$R$ is {\em locally finite} if the sum (\ref{rdef}) consists of finitely many terms, i.e.,
$R^{a,b,c}_{i,j,k}=0$ for all but finitely many $(a,b,c)$ for any given $(i,j,k)$.

\section{Solutions of $RLLL$ relation for $L = L^Z$ and $L^O$}\label{sec:ZO}

In this section we treat the cases in which 
$L_{124}, L_{135}$ and $L_{236}$ are chosen as $L^Z$or $L^O$ independently.
It turns out that they always admit a unique $R$ up to normalization 
in a sector specified by appropriate parity conditions.
Their explicit forms will be presented case by case.
We write the characteristic function as 
$\theta(\text{true}) =1, \theta(\text{false})=0$, 
$\delta^a_b = \theta(a=b)$ and use the following notation:
\begin{align}
&(z;q)_m = \frac{(z;q)_\infty}{(zq^m;q)_\infty},
\quad 
(z;q)_\infty = \prod_{n\ge 0}(1-zq^n),
\quad
\binom{n}{m}_{\!q} = \frac{(q;q)_n}{(q;q)_m(q;q)_{n-m}},
\label{qfac}\\
&
{}_2\phi_1\left({\alpha, \beta  \atop \gamma }; q, z\right)
= \sum_{n \ge 0} \frac{(\alpha;q)_n(\beta;q)_n }{(\gamma;q)_n(q;q)_n}z^n.
\label{qhyp}
\end{align}
The above convention for 
$(z;q)_m$ valid for any  $m \in \Z$ is standard and essential 
in the working below.
In particular $1/(q;q)_a =0$ for $a \in \Z_{<0}$, and we will freely use 
$(z;q)_m = 1/(zq^m;q)_{-m}$ and 
$(z;q)_m/(z;q)_n = (zq^n;q)_{m-n}$, etc.
The $q$-binomial $\binom{n}{m}_{\!q}$ is zero unless $0 \le m \le n$.
The $q$-hypergeometric series will always appear in the terminating situation, i.e., 
$\alpha$ or $\beta\in q^{\Z_{\le 0}}$.

\subsection{ZZZ type}

We consider the $RLLL$ relation
\begin{align}\label{ZZZ}
R_{456}L^{Z}_{236}L^{Z}_{135}L^{Z}_{124}
=L^{Z}_{124}L^{Z}_{135}L^{Z}_{236}R_{456},
\end{align}
where $L^{Z}_{124}, L^{Z}_{135}, L^{Z}_{236}$ are given by (\ref{LW}) with 
$(r,s,t,w)=(r_1, s_1,t_1,w_1)$, $(r_2, s_2,t_2,w_2)$, $(r_3, s_3,t_3,w_3)$.
In this case, 
$R \in \mathrm{End}(F \otimes F \otimes F)$ and  the sum (\ref{rdef}) 
extends over $a,b,c \in \Z$.
The equality (\ref{ZZZ}) holds in  
$\mathrm{End}(V\otimes V \otimes V \otimes F \otimes F \otimes F)$.

The 18 equations (\ref{qybe}) corresponding to (\ref{ZZZ}) have been listed 
in Appendix \ref{app:eqs}.
As an illustration consider the cases $(a,b,c,i,j,k) =$ (0,0,1,0,0,1), (1,0,0,1,0,0), 
(1,0,0,0,0,1), (1,1,0,0,1,1), (1,0,1,0,1,1) and (1,1,0,1,0,1):
\begin{align}
&R(1\otimes X\otimes X) = (1\otimes X \otimes X)R,
\quad 
R(X \otimes X \otimes 1) = (X \otimes X \otimes 1)R,
\\
&-r_1r_3 R(1\otimes Z \otimes 1) =  (q t_1t_3w_1 X \otimes Z \otimes X -r_2 Z \otimes 1 \otimes Z)R,
\\
&R(-qt_1t_3w_3 X \otimes Z \otimes X + s_2 Z \otimes 1 \otimes Z) = s_1s_3 (1 \otimes Z \otimes 1)R,  
\\
&t_1 R (X \otimes Z \otimes Z^{-1}(r_3 s_3 - t_3^2 w_3 X^2) + s_2 t_3 Z \otimes 1 \otimes X) = s_3 t_2 (Z \otimes X \otimes 1)R, \\
& R (t_3 w_3Z^{-1}(r_1 s_1 - t_1^2 w_1 X^2) \otimes Z \otimes X + s_2 t_1 w_1 X \otimes 1 \otimes Z) = s_1 t_2 w_2 (1 \otimes X \otimes Z)R.
\end{align}
Taking their matrix elements for the transition
$|i\rangle \otimes |j\rangle \otimes |k\rangle \mapsto |a\rangle \otimes |b\rangle \otimes |c\rangle$,
we get the recursion relations for elements of $R$:
\begin{align}
&R^{a,b,c}_{i,j-1,k-1} =R^{a,b+1,c+1}_{i,j,k},
\quad 
R^{a,b,c}_{i-1,j-1,k} =R^{a+1,b+1,c}_{i,j,k},
\label{sm:rec1}\\
&(q^{a+c}r_2 - q^j r_1r_3)R^{a,b,c}_{i,j,k} = q^{1+b}t_1t_3w_1R^{a+1,b,c+1}_{i,j,k},
\label{sm:rec2}\\
&(q^{i+k}s_2-q^bs_1s_3)R^{a,b,c}_{i,j,k}= q^{1+j}t_1t_3w_3R^{a,b,c}_{i-1,j,k-1},
\label{sm:rec3}\\
& q^j r_3 s_3 t_1 R^{a,b,c}_{i-1,j,k} - q^{j+2} t_1 t_3^2 w_3 R^{a,b,c}_{i-1,j,k-2} + q^{i+k} s_2 t_3 R^{a,b,c}_{i,j,k-1} = q^{a+k} s_3 t_2 R^{a,b+1,c}_{i,j,k}, 
\label{sm:rec4}\\
& q^j r_1 s_1 t_3 w_3 R^{a,b,c}_{i,j,k-1} - q^{j+2} t_1^2 t_3w_1  w_3 R^{a,b,c}_{i-2,j,k-1} + q^{i+k} s_2 t_1 w_1 R^{a,b,c}_{i-1,j,k} = q^{c+i} s_1 t_2 w_2 R^{a,b+1,c}_{i,j,k}.
\label{sm:rec5}
\end{align}
Each recursion relation is actually a collection of infinitely many linear equations on infinitely many $R^{a,b,c}_{i,j,k}$'s  
depending on the choice of $(a,b,c,i,j,k) \in \Z^6$.

Given two integers $d$ and $d'$, we write the pair 
$(d\;\mathrm{mod\, 2},d'\;\mathrm{mod\, 2}) \in \Z_2 \times \Z_2$
simply as $(d,d')_{\mathrm{mod \, 2}}$.

\begin{proposition}\label{pr:d12}
(i) Any recursion relation  
consists of only those 
$R^{a,b,c}_{i,j,k}$'s having the same parity pair $(d_1, d_2)_{\mathrm{mod \, 2}}$, where 
 $d_1=a+c-j$ and $d_2=b-i-k$.
(ii) Each subsystem of recursion relations corresponding to a given $(d_1, d_2)_{\mathrm{mod \, 2}}$
allows a solution of dimension at most one.
\end{proposition}
\begin{proof}
Claim (i) can be checked directly. 
Let us prove Claim (ii). 
First, we reduce $b, c$ and $k$ to $0$ by using (\ref{sm:rec1}) and (\ref{sm:rec3}).
The result reads
\begin{equation}
    R^{a,b,c}_{i,j,k} = q^{(c+i-j)(c-k)} \left(\frac{t_1 t_3 w_3}{s_2}\right)^{-c+k} \frac{1}{(q^{b-i-k}\frac{s_1 s_3}{s_2}; q^2)_{-c+k}} R^{a-b+c,0,0}_{i-k-b+2c,j-b,0}.
    \label{sm:prd12_1}
\end{equation}
Applying this to (\ref{sm:rec4}) and (\ref{sm:rec5}) with $b=c=k=0$ we get 
\begin{align}
&q^j r_3 t_1^2 w_3 R^{a,0,0}_{i-1,j,0} + q^{-j} s_1 (q^{i+1} s_2 - s_1 s_3) R^{a,0,0}_{i+1,j,0} 
= q^a t_1 t_2 w_3 R^{a-1,0,0}_{i-1,j-1,0},
\label{sm:prd12_2}
\\
&q^2 s_3 t_1^2 w_1 R^{a,0,0}_{i-1,j,0} + r_1(q^{i+1}s_2 - s_1 s_3) R^{a,0,0}_{i+1,j,0} 
= q^{i+1} t_1 t_2 w_2 R^{a-1,0,0}_{i-1,j-1,0}.
\label{sm:prd12_3}
\end{align}
Eliminating $R^{a,0,0}_{i+1,j,0}$ here leads to the recursion relation
\begin{equation}
 R^{a,0,0}_{i,j,0} = q^i \frac{t_2 w_2}{s_3 t_1 w_1} 
\frac{1 - q^{a-i+j-2}\frac{r_1 w_3}{s_1 w_2}}{1 - q^{2j-2} \frac{r_1 r_3 w_3}{s_1 s_3 w_1}}
 R^{a-1,0,0}_{i,j-1,0}.
\label{sm:prd12_4}
\end{equation}
We remark that combination of (\ref{sm:prd12_1}) and (\ref{sm:prd12_4}) allows one to express 
$R^{a,b,c}_{i,j,k}$ in terms of $R^{0,0,0}_{i-k-b+2c,j-a-c,0}$ whose indices satisfy 
$i-k-b+2c \equiv d_2$ and $j-a-c \equiv d_1$ mod 2.

Next, consider (\ref{sm:rec2}) and  (\ref{sm:rec4}) again with $a=b=c=k=0$.
Reducing them to the relations among $R^{0,0,0}_{\bullet, \bullet, 0}$ by the above remark, and 
taking a suitable combination, we get 
\begin{align}
 R^{0,0,0}_{i,j,0} 
 &= q^{2+i} \frac{t_2^2w_2} {r_2s_1 s_3} \frac{(1-q^{j-2+i} \frac{r_3 w_2}{s_3 w_1}) (1-q^{j-2-i} \frac{r_1 w_3}{s_1 w_2})}{(1-q^j \frac{r_1 r_3}{r_2}) (1-q^{2j-2} \frac{r_1 r_3 w_3}{s_1 s_3 w_1}) (1-q^{2j-4} \frac{r_1 r_3 w_3}{s_1 s_3 w_1})} R^{0,0,0}_{i,j-2,0}, 
\label{sm:prd12_6}
\\    
R^{0,0,0}_{i,j,0} &= q^{-2i+j+2} \frac{s_3 t_1^2 w_1 w_3}{s_1 s_2 w_2} \frac{1 - q^{i+j-2} \frac{r_3w_2}{s_3w_1}}{(1-q^{-i}\frac{s_1s_3}{s_2})(1-q^{-i+j}\frac{r_1w_3}{s_1w_2})} R^{0,0,0}_{i-2,j,0}.
\label{sm:prd12_5}
\end{align}
Thus we find any $R^{a,b,c}_{i,j,k}$ is uniquely expressed as $R^{0,0,0}_{p_2,p_1,0}$ times 
known factors, where $p_1, p_2 \in \{0,1\}$ are determined by 
$p_1\equiv d_1, p_2 \equiv d_2$ mod 2.
\end{proof}

For $a,b,c,i,j,k \in \Z$ set 
\begin{align}
\begin{split}
R^{a,b,c}_{i,j,k}&= 
\left(\frac{r_2}{t_1t_3w_1}\right)^{\frac{d_1}{2}}
\left(\frac{s_2}{t_1t_3w_3}\right)^{\frac{d_2}{2}}
\left( \frac{t_2}{s_1t_3}\right)^{\frac{d_3}{2}}
\left(\frac{t_2w_2}{s_3t_1w_1}\right)^{\frac{d_4}{2}}
\\
&\quad \times q^{\varphi}
\frac{\Phi_{d_2}\left(\frac{s_1s_3}{s_2}\right)
\Phi_{d_3}\left(\frac{r_3w_2}{s_3w_1}\right)
\Phi_{d_4}\left(\frac{r_1w_3}{s_1w_2}\right)}
{\Phi_{-d_1}\left(\frac{q^2r_1r_3}{r_2}\right)\Phi_{d_3+d_4}\left(\frac{r_1r_3w_3}{s_1s_3w_1}\right)},
\end{split}
\label{rwww}
\\
\varphi &= \frac{1}{4}\bigl((d_1-d_2)(d_1+d_2+d_3+d_4)+d_3d_4\bigr)-d_1,
\label{phi}\\
\begin{pmatrix}
d_1\\ d_2
\end{pmatrix}&=
\begin{pmatrix}
a+c-j \\  b-i-k
\end{pmatrix},\quad
\begin{pmatrix}
d_3\\ d_4
\end{pmatrix}=
\begin{pmatrix} -a-b+c+i+j-k 
\label{ds}\\
a-b-c-i+j+k
\end{pmatrix},
\\
\Phi_m(z)&= \frac{1}{(zq^m;q^2)_\infty}\quad (m \in \Z),
\label{pm}
\end{align}
where $d_1$ and $d_2$ are the same as those in Proposition \ref{pr:d12}.
It is easy to see $\varphi \in \Z+ (d_1-1)d_2/2$.
The dependence on $t_1,t_2,t_3$ is actually 
by the combination $t_1^{-a+i}t_2^{-b+j}t_3^{-c+k}$, 
which corresponds to a similarity transformation.

By Proposition \ref{pr:d12}, we know that the solution $R$ of (\ref{ZZZ}), if exists, 
is unique up to normalization in each sector specified by $(d_1, d_2)_{\mathrm{mod \, 2}}$.
The following result establishes the existence together with an explicit form.

\begin{theorem}\label{th:www}
The 3D $R$ defined by (\ref{rwww})--(\ref{pm})  
satisfies the $RLLL$ relation (\ref{ZZZ}).
\end{theorem}
\begin{proof}
From Proposition \ref{pr:d12} and $d_3\equiv d_4 \equiv d_1+d_2\mod 2$,
the replacement
\begin{align}\label{pt}
\Phi_m(z) \rightarrow {\tilde \Phi}_m(z) = 
\begin{cases}
(z;q^2)_\infty/(zq^m;q^2)_\infty  = (z;q^2)_{m/2} & (m \in 2\Z),
\\
(zq;q^2)_\infty/(zq^m;q^2)_\infty = (zq;q^2)_{(m-1)/2} & (m \in 2\Z+1)
\end{cases}
\end{align}
changes the individual recursion relations only by an overall scalar.
The results become the relations among 
finitely many rational functions.
To check them is straightforward.
\end{proof}
 
As the above proof indicates, one may just postulate the property 
\begin{align}\label{prec}
\Phi_{m+2}(z) =  (1-zq^m)\Phi_m(z)
\end{align}
instead of specifying $\Phi_m(z)$ concretely as (\ref{pm}).
Another option of such sort is to make the replacement 
\begin{align}
1/\Phi_{-d_1}\left(\frac{q^2r_1r_3}{r_2}\right)
\rightarrow
q^{-\frac{d_1^2}{4}+\frac{d_1}{2}}
\left(-\frac{r_1r_3}{r_2}\right)^{\frac{d_1}{2}}
\Phi_{d_1}\left(\frac{r_2}{r_1r_3}\right),
\end{align}
which makes the formula (\ref{rwww}) 
more symmetric with respect to $d_1$ and $d_2$
at the cost of the appearance of the factor $(-1)^{d_1/2}$.
The $R$ is not locally finite.
From (\ref{Linv}), its inverse is given by 
\begin{align}\label{zzzinv}
R^{-1} = 
(\text{scalar}) 
R\left.\right|_{r_i \leftrightarrow s_i, 
t_i \rightarrow t_iw_i,w_i\rightarrow w^{-1}_i\,(i=1,2,3)}.
\end{align}

The parity condition on $(d_1, d_2)$ 
mixes the indices $i,j,k$ labeling incoming states and $a,b,c$ concerning outgoing ones.
See (\ref{rdef}). 
To illustrate the resulting sectors, we introduce the subspace 
\begin{align}\label{Fp}
\mathcal{F}_{p_1,p_2} = \bigoplus_{i+k \equiv p_1,\, j \equiv p_2 \,
\mathrm{mod}\, 2} 
\C |i \rangle \otimes |j \rangle \otimes | k \rangle \subset F ^{\otimes 3}
\qquad (p_1, p_2 = 0,1).
\end{align}
From the proof of Proposition \ref{pr:d12}, 
the solution space of $R$ is four dimensional whose basis corresponds to 
the ``initial condition" of the recursion relation taken as 
$(R^{0,0,0}_{0,0,0},R^{0,0,0}_{1,0,0},R^{0,0,0}_{0,1,0},R^{0,0,0}_{1,1,0}) =
(1,0,0,0)$, $(0,1,0,0)$, $(0,0,1,0)$, $(0,0,0,1)$.
Call them $R[0,0], R[0,1], R[1,0], R[1,1]$ respectively so that $R[p_1,p_2]$ 
is the base corresponding to the choice $R^{0,0,0}_{p_2,p_1,0} = 1$ 
according to the remark after (\ref{sm:prd12_4}).\footnote{The 
formula (\ref{rwww}) has not been so normalized.}
Then they act on (\ref{Fp}) as in Figure \ref{fig:Rpp}.

\begin{figure}[H]
\begin{picture}(400,110)(-5,-10)

\put(0,80){
\put(61,-21){$R[0,0]$}
\put(0,0){$\mathcal{F}_{0,0}$} \put(22,-2){\begin{LARGE}\rotatebox{90}{$\circlearrowleft$}\end{LARGE}}
\put(51,0){$\mathcal{F}_{0,1} \rightleftarrows \mathcal{F}_{1,0}$}
\put(120,0){\put(0,0){$\mathcal{F}_{1,1}$} \put(22,-2){\begin{LARGE}\rotatebox{90}{$\circlearrowleft$}\end{LARGE}}}
}

\put(220,80){
\put(61,-21){$R[1,1]$}
\put(0,0){$\mathcal{F}_{0,1}$} \put(22,-2){\begin{LARGE}\rotatebox{90}{$\circlearrowleft$}\end{LARGE}}
\put(51,0){$\mathcal{F}_{0,0} \rightleftarrows \mathcal{F}_{1,1}$}
\put(120,0){\put(0,0){$\mathcal{F}_{1,0}$} \put(22,-2){\begin{LARGE}\rotatebox{90}{$\circlearrowleft$}\end{LARGE}}}
}

\put(10,10){
\put(52,-20){$R[1,0]$}
\put(-3,3){\line(-1,0){10}}\put(-13,3){\line(0,1){15}}
\put(-13,18){\line(1,0){82}}
\put(0,0){$\mathcal{F}_{0,0}  \longrightarrow
\mathcal{F}_{1,0}  \longrightarrow
\mathcal{F}_{1,1}  \longrightarrow \mathcal{F}_{0,1}$}
\put(150,18){\vector(-1,0){82}}
\put(140,3){\line(1,0){10}}\put(150,3){\line(0,1){15}}
}

\put(230,10){
\put(52,-20){$R[0,1]$}
\put(-3,3){\line(-1,0){10}}\put(-13,3){\line(0,1){15}}
\put(-13,18){\line(1,0){82}}
\put(0,0){$\mathcal{F}_{0,0}  \longrightarrow
\mathcal{F}_{0,1}  \longrightarrow
\mathcal{F}_{1,1}  \longrightarrow \mathcal{F}_{1,0}$}
\put(150,18){\vector(-1,0){82}}
\put(140,3){\line(1,0){10}}\put(150,3){\line(0,1){15}}
}

\end{picture}
\caption{Action of the four fundamental solutions 
$R[0,0], R[1,0],R[0,1],R[1,1]$ on the subspaces $\mathcal{F}_{p_1, p_2}$ defined in (\ref{Fp}).
For example in $R[1,0]$, the condition $(d_1,d_2)=(a+c-j,b-i-k) \equiv (1,0)$ on 
$|i\rangle \otimes |j\rangle \otimes |k\rangle \mapsto |a\rangle \otimes |b\rangle \otimes |c\rangle$
enforces $R[1,0] \mathcal{F}_{0,0}  \subseteq \mathcal{F}_{1,0}$,
$R[1,0] \mathcal{F}_{1,0}  \subseteq \mathcal{F}_{1,1}$,
$R[1,0] \mathcal{F}_{1,1}  \subseteq \mathcal{F}_{0,1}$ and 
$R[1,0] \mathcal{F}_{0,1}  \subseteq \mathcal{F}_{0,0}$.}
\label{fig:Rpp}
\end{figure}
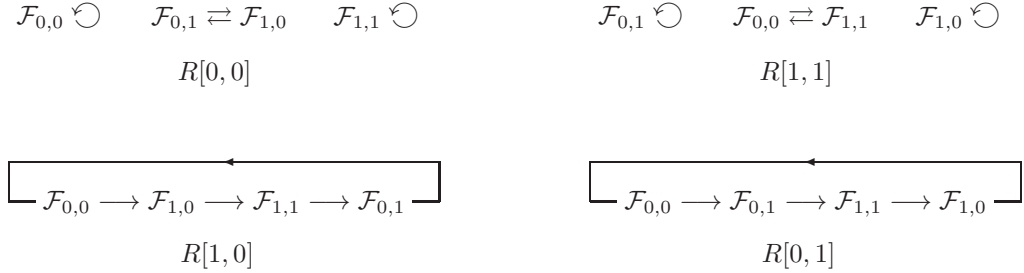

Similar decompositions according to a parity condition also take place in the forthcoming 
Theorems \ref{th:ooz}, \ref{th:zoo}, \ref{th:owo}, 
\ref{th:xxz}, \ref{th:zxx} and \ref{th:xzx}.

\begin{remark}
Let $L^{Z_\pm}$ be the 3D $L$ operator (\ref{LW}) with $\pi_Z$ in (\ref{piz}) replaced by 
\begin{align}\label{xz2}
X|m\rangle = |m\mp 1\rangle,\quad X^{-1}|m\rangle = |m \pm 1\rangle,
\quad  Z|m\rangle = q^{\pm m}|m\rangle, \quad  Z^{-1}|m\rangle = q^{ \mp m}|m\rangle.
\end{align}
Theorem \ref{th:www} is concerned with $L^{Z} = L^{Z_+}$.
Consider a variant of (\ref{ZZZ}) given by 
\begin{align}\label{WWW2}
R(\varepsilon_1, \varepsilon_2, \varepsilon_3)_{456}L^{Z_{\varepsilon_3}}_{236}
L^{Z_{\varepsilon_2}}_{135}L^{Z_{\varepsilon_1}}_{124}
=L^{Z_{\varepsilon_1}}_{124}
L^{Z_{\varepsilon_2}}_{135}L^{Z_{\varepsilon_3}}_{236}
R(\varepsilon_1, \varepsilon_2, \varepsilon_3)_{456}
\qquad (\varepsilon_1, \varepsilon_2, \varepsilon_3 \in \{-1,1\}).
\end{align}
 Then elements of $R(\varepsilon_1, \varepsilon_2, \varepsilon_3)$ is given by 
 \begin{align}\label{rwww2}
 R(\varepsilon_1, \varepsilon_2, \varepsilon_3)^{a,b,c}_{i,j,k} 
 = R^{\varepsilon_1a, \,\varepsilon_ 2 b, \,\varepsilon_3 c}_{\varepsilon_1i, \, \varepsilon_ 2 j, \,\varepsilon_3 k},
 \end{align}
where the RHS is defined by (\ref{rwww})--(\ref{pm}) which corresponds to 
$R(+,+,+)$.
\end{remark}

\subsection{OZZ type}

We consider the $RLLL$ relation
\begin{align}\label{OWW}
R_{456}L^{Z}_{236}L^{Z}_{135}L^{O}_{124}
=L^{O}_{124}L^{Z}_{135}L^{Z}_{236}R_{456},
\end{align}
where $L^{Z}_{135}$ and $L^{Z}_{236}$ are given by (\ref{LW}) with 
$(r,s,t,w)=(r_2, s_2,t_2,w_2)$ and $(r_3, s_3,t_3,w_3)$, respectively.
In this case, 
$R \in \mathrm{End}(F_+\otimes F \otimes F)$ and  the sum (\ref{rdef}) 
extends over $a \in \Z_{\ge 0}$ and $b,c \in \Z$.
The equality (\ref{OWW}) holds in  
$\mathrm{End}(V\otimes V \otimes V \otimes F_+ \otimes F \otimes F)$.

Here are some examples of the $RLLL$ relation (\ref{OWW}):
\begin{align}
&R^{a,b,c}_{i,j-1,k-1}=R^{a,b+1,c+1}_{i,j,k},
\quad
q^{i}R^{a,b,c}_{i,j-1,k}= q^a R^{a,b+1,c}_{i,j,k},
\\
&q^{a+c}r_2R^{a,b,c}_{i,j,k} - q^k\mu t_2R^{a,b,c}_{i,j-1,k}+q^b(1-q^{2a+2})\mu t_3R^{a+1,b,c+1}_{i,j,k}=0,
\\
&q^cr_2R^{a-1,b,c}_{i,j,k}-q^jr_3R^{a,b,c}_{i,j,k}-q^{1+a+b}\mu t_3R^{a,b,c+1}_{i,j,k}=0,
\label{3t}
\\
&q^{i+j}r_3s_3R^{a,b,c}_{i,j,k} +  q^k \mu s_2t_3 R^{a,b,c}_{i+1,j,k-1}
- q^k\mu s_3t_2 R^{a-1,b+1,c}_{i,j,k} - q^{2+i+j}t_3^2w_3R^{a,b,c}_{i,j,k-2} = 0.
\end{align}
The boundary condition 
\begin{align}\label{bc1}
R^{a,b,c}_{i,j,k} = 0 \quad \text{if} \;\min(a,i)<0.
\end{align}
has to be taken into account.
Thus for example when $a=0$,  (\ref{3t}) is to be understood as 
 $q^jr_3R^{0,b,c}_{i,j,k}+q^{1+b}\mu t_3R^{0,b,c+1}_{i,j,k}=0$.

For $a,b,c,i,j,k \in \Z$, set
\begin{align}\label{roww}
\begin{split}
R^{a,b,c}_{i,j,k} &= 
\left(\frac{r_2}{r_3}\right)^a
\left(\frac{s_3}{s_2}\right)^i
\left(\frac{t_2w_2}{\mu s_2}\right)^{-b+j}
\left(-\frac{\mu t_3}{r_3}\right)^{-c+k}\frac{1}{(q^2;q^2)_a}q^{(a - b + j-1) c - (i - b + j-1) k-aj + bi}
\\
& \quad \times \sum_{\beta=0}^i q^{\beta (\beta + 2 j - 2 b - 1)} 
\left(-y\right)^{\beta}
\binom{i}{\beta}_{\!q^2}
\Bigl(xq^{2k-2c-2\beta+2};q^2\Bigr)_a,
\end{split}
\\
x &= \frac{\mu^2 s_2}{r_2w_2}, \qquad y = \frac{r_3w_3}{\mu^2 s_3},
\qquad z = xq^{2k-2c+2}.
\label{xy}
\end{align}
For the convenience of the proof of Theorem \ref{th:oww}, 
we have enlarged the range of the indices $a$ and $i$ from 
$\Z_{\ge 0}$ to $\Z$.
The property (\ref{bc1}) is satisfied thanks to the 
factor $\binom{i}{\beta}_{q^2}/(q^2;q^2)_a$.
The formula (\ref{roww})  is also presented as a terminating $q$-hypergeometric series:
\begin{equation}\label{roww2}
\begin{split}
R^{a,b,c}_{i,j,k} &= \theta(i\ge 0)
\left(\frac{r_2}{r_3}\right)^a
\left(\frac{s_3}{s_2}\right)^i
\left(\frac{t_2w_2}{\mu s_2}\right)^{-b+j}
\left(-\frac{\mu t_3}{r_3}\right)^{-c+k}
\frac{(z;q^2)_a}{(q^2;q^2)_a}q^{(a - b + j-1) c - (i - b + j-1) k-aj + bi}
\\
& \quad \times 
{}_2\phi_1\left({q^{-2i}, z^{-1}q^2 \atop z^{-1}q^{-2a+2}}; q^2, yq^{2i+2j-2a-2b}\right).
\end{split}
\end{equation}

\begin{theorem}\label{th:oww}
The $RLLL$ relation (\ref{OWW}) has a unique solution $R$ up to normalization.
It is given by (\ref{roww})--(\ref{roww2}).
\end{theorem}
\begin{proof}
The first claim, i.e., uniqueness, 
can be shown by an argument similar to Proposition \ref{pr:d12}.
To prove the second claim, 
let $S^a_{i,j-b,k-c}(x,y)$ denote the second line of (\ref{roww}).
One sees that 
$S^a_{i,j,k}(x,y) = \sum_{\beta = 0}^i (-y)^\beta S^a_{i,j,k,\beta}(x)$, 
where $S^a_{i,j,k,\beta}(x) = q^{\beta(\beta+2j-1)}\binom{i}{\beta}_{q^2}
(xq^{2k-2\beta+2};q^2)_a$
is a polynomial in $x$ and $y$.
The equation  (\ref{OWW}) is reduced to the recursion relations among 
$S^a_{i,j,k}(x,y)$ with coefficients including $q,q^a,q^i,q^j,q^k,x, y$ only.
By picking the coefficients of $y^\beta$, 
they are reduced to the relations containing finitely many 
$S^a_{i,j,k,\beta}(x)$'s. To check them is straightforward.
This proves the recursion relations for generic $a$ and $i$.
This fact together with (\ref{bc1}) assure that they are also valid 
in the vicinity of $a=0$ and $i=0$.
\end{proof}

As for the last point of the proof, 
a similar and more detailed explanation is available 
in the proof of Theorem \ref{th:ooz}.
The $R$ is not locally finite.

\subsection{ZZO type}

We consider the $RLLL$ relation
\begin{align}\label{WWO}
R_{456}L^{O}_{236}L^{Z}_{135}L^{Z}_{124}
=L^{Z}_{124}L^{Z}_{135}L^{O}_{236}R_{456},
\end{align}
where $L^{Z}_{124}$ and $L^{Z}_{135}$ are given by (\ref{LW}) with 
$(r,s,t,w)= (r_1, s_1,t_1,w_1)$ and $(r_2, s_2,t_2,w_2)$, respectively.
In this case, 
$R \in \mathrm{End}(F\otimes F \otimes F_+)$ and  the sum (\ref{rdef}) 
extends over $a,b\in \Z$ and $c\in \Z_{\ge 0}$.
The equality (\ref{WWO}) holds in  
$\mathrm{End}(V\otimes V \otimes V \otimes F \otimes F \otimes F_+)$.

Here are some examples of the $RLLL$ relation (\ref{WWO}):
\begin{align}
&R^{a,b,c}_{i-1,j-1,k} = R^{a+1,b+1,c}_{i,j,k},\quad 
q^k R^{a,b,c}_{i,j-1,k} = q^c R^{a,b+1,c}_{i,j,k},
\\
&q^{a+c}\mu r_2R^{a,b,c}_{i,j,k} - q^i t_2w_2R^{a,b,c}_{i,j-1,k}
+ q^b(1-q^{2c+2})t_1w_1R^{a+1,b,c+1}_{i,j,k} = 0,
\\
&q^a \mu r_2R^{a,b,c-1}_{i,j,k} -q^j \mu r_1R^{a,b,c}_{i,j,k}
-q^{1+b+c}t_1w_1R^{a+1,b,c}_{i,j,k} = 0,
\\
&q^{j+k}\mu r_1s_1 R^{a,b,c}_{i,j,k} - q^is_1t_2w_2 R^{a,b+1,c-1}_{i,j,k}
-q^{2+j+k} \mu  t_1^2w_1R^{a,b,c}_{i-2,j,k} + q^i s_2t_1w_1R^{a,b,c}_{i-1,j,k+1}=0.
\end{align}
One has the boundary condition analogous to (\ref{bc1}):
\begin{align}\label{bc2}
R^{a,b,c}_{i,j,k} = 0 \quad \text{if} \;\min(c,k)<0.
\end{align}

For $a, b, c, i, j, k \in \Z_{\ge 0}$, set
\begin{align}\label{rwwo}
\begin{split}
R^{a,b,c}_{i,j,k} &= 
\left(\frac{r_2}{r_1}\right)^c
\left(\frac{s_1}{s_2}\right)^k
\left(\frac{\mu t_2}{s_2}\right)^{-b+j}
\left(-\frac{t_1w_1}{\mu r_1}\right)^{-a+i}
\frac{1}{(q^2;q^2)_c}q^{(c - b + j-1) a - (k - b + j-1)i -cj + bk}
\\
& \quad \times \sum_{\beta=0}^k q^{\beta (\beta + 2 j - 2 b - 1)} 
\left(-y\right)^{\beta}
\binom{k}{\beta}_{\!q^2}
\Bigl(xq^{2i-2a-2\beta+2};q^2\Bigr)_c,
\end{split}
\\
x &= \frac{s_2w_2}{\mu^2 r_2},\qquad y = \frac{\mu^2 r_1}{s_1w_1},
\qquad z = x q^{2i-2a+2},
\label{xy2}
\end{align}
where we have redefined $x,y,z$ changing (\ref{xy}).
It is also presented as a terminating $q$-hypergeometric series:
\begin{equation}\label{rwwo2}
\begin{split}
R^{a,b,c}_{i,j,k} &= \theta(k \ge 0)
\left(\frac{r_2}{r_1}\right)^c
\left(\frac{s_1}{s_2}\right)^k
\left(\frac{\mu t_2}{s_2}\right)^{-b+j}
\left(-\frac{t_1w_1}{\mu r_1}\right)^{-a+i}
\frac{(z;q^2)_c}{(q^2;q^2)_c}q^{(c - b + j-1) a - (k - b + j-1)i -cj + bk}
\\
& \quad \times 
{}_2\phi_1\left({q^{-2k}, z^{-1}q^2 \atop z^{-1}q^{-2c+2}}; q^2, yq^{2j+2k-2b-2c}\right).
\end{split}
\end{equation}

\begin{theorem}\label{th:wwo}
The $RLLL$ relation (\ref{WWO}) has a unique solution $R$ up to normalization.
It is given by (\ref{rwwo})--(\ref{rwwo2}).
\end{theorem}
The proof is similar to Theorem \ref{th:oww}.
The $R$ is not locally finite.

\subsection{ZOZ type}\label{ss:wow}
We consider the $RLLL$ relation
\begin{align}\label{WOW}
R_{456}L^{Z}_{236}L^{O}_{135}L^{Z}_{124}
=L^{Z}_{124}L^{O}_{135}L^{Z}_{236}R_{456},
\end{align}
where 
$L^{Z}_{124}$ and $L^{Z}_{236}$ are given by (\ref{LW}) with 
$(r,s,t,w)= (r_1, s_1,t_1,w_1)$ and $(r_3, s_3,t_3,w_3)$, respectively.
In this case, 
$R \in \mathrm{End}(F\otimes F_+ \otimes F)$ and  the sum (\ref{rdef}) 
extends over $a,c\in \Z$ and $b\in \Z_{\ge 0}$.
The equality (\ref{WOW}) holds in  
$\mathrm{End}(V\otimes V \otimes V \otimes F \otimes F_+ \otimes F)$.

Here are some examples of the $RLLL$ relation (\ref{WOW}):
\begin{align}
&q^jR^{a,b,c}_{i,j,k-1} = q^bR^{a,b,c+1}_{i,j,k},\quad 
q^j R^{a,b,c}_{i-1,j,k} = q^b R^{a+1,b,c}_{i,j,k},
\\
&q^{a+c}R^{a,b,c}_{i,j,k} - r_1r_3R^{a,b,c}_{i,j+1,k}
- qt_1t_3w_1R^{a+1,b-1,c+1}_{i,j,k} = 0,
\\
&q^{i+k}R^{a,b,c}_{i,j,k} -s_1s_3R^{a,b-1,c}_{i,j,k}
-qt_1t_3w_3R^{a,b,c}_{i-1,j+1,k-1} = 0,
\\
&q^{a+j+k}r_1 R^{a,b,c}_{i,j,k} - \mu r_1s_1t_3R^{a,b-1,c+1}_{i,j,k}
-  q^{a+c}\mu t_1R^{a+1,b,c}_{i,j,k} + \mu t_1^2 t_3 w_1R^{a+2,b-1,c+1}_{i,j,k} = 0.
\end{align}
The boundary condition is given by 
\begin{align}\label{bc3}
R^{a,b,c}_{i,j,k} = 0 \quad \text{if} \;\min(b,j)<0.
\end{align}

For $a, b, c, i, j, k \in \Z_{\ge 0}$, set
\begin{align}
\begin{split}
R^{a,b,c}_{i,j,k} &= \theta(j \ge 0)
\frac{(s_1s_3)^b}{(r_1r_3)^j}
\left(\frac{r_1}{\mu t_1}\right)^{a-i}
\left(\frac{\mu r_3}{t_3w_3}\right)^{c-k}
\frac{1}{(q^2;q^2)_b}
q^{(j - b) (a + c) + b (a + c - i - k) - (i - a) (k - c)}
\\
&\times 
\sum_{\beta=0}^b q^{\beta(\beta+2i-2a+1)}
\left(-y\right)^\beta
\binom{b}{\beta}_{\! q^2}
\left(q^{2j+2k-2c-2\beta}x^{-1};q^{-2}\right)_\beta
\left(q^{2k-2c-2\beta+2}x^{-1};q^2\right)_{b-\beta},
\end{split}
\label{rwow}\\
x &= \frac{\mu^2 s_1}{r_1w_1},\quad 
\quad  y = \frac{\mu^2r_3}{s_3w_3}.
\label{xy3}
\end{align}
This can also be expressed as a terminating series 
similar to a generalized $q$-hypergeometric ${}_3\phi_2$:
\begin{align}
\begin{split}
R^{a,b,c}_{i,j,k} &= \theta(j \ge 0)
\frac{(s_1s_3)^b}{(r_1r_3)^j}
\left(\frac{r_1}{\mu t_1}\right)^{a-i}
\left(\frac{\mu r_3}{t_3w_3}\right)^{c-k}
\frac{(q^{2-2c+2k}x^{-1},q^2)_b}{(q^2;q^2)_b}
q^{(j - b) (a + c) + b (a + c - i - k) - (i - a) (k - c)}
\\
&\times \sum_{\beta=0}^b
(q^{2i+2j-2a-2b+2}y)^\beta
\frac{(q^{-2b};q^2)_\beta(q^{2c-2k}x;q^2)_\beta(q^{2c-2j-2k}x;q^2)_{2\beta}}
{(q^2;q^2)_\beta (q^{-2b+2c-2k}x;q^2)_{2\beta}(q^{2c-2j-2k}x;q^2)_{\beta}}.
\end{split}
\label{rwow2}
\end{align}
The difference from ${}_3\phi_2$ is the factors $(\bullet;q^2)_{2\beta}$.

\begin{theorem}\label{th:wow}
The $RLLL$ relation (\ref{WOW}) has a unique solution $R$ up to normalization.
It is given by (\ref{rwow})--(\ref{xy3}).
\end{theorem}
The proof is similar to Theorem \ref{th:oww}.
The $R$ is not locally finite.

\subsection{OOZ type}\label{ss:oow}

We consider the $RLLL$ relation
\begin{equation}\label{OOW}
R_{456} L_{236}^ZL_{135}^O L_{124}^O= L_{124}^OL_{135}^OL_{236}^Z R_{456}, 
\end{equation}
where $L_{124}^O$ and $L_{135}^O$ are given by (\ref{LO}) with $\mu = \mu_1$ 
and $\mu_2$, respectively, and $L_{236}^Z$ is given by (\ref{LW}) with $(r, s, t, w) = (r_3, s_3, t_3, w_3)$. 
In this case, $R \in \mathrm{End}(F_+ \otimes F_+ \otimes F)$ 
and the sum (\ref{rdef}) extends over $a, b \in \mathbb{Z}_{\geq 0}$ and $c \in \mathbb{Z}$. 
The equality (\ref{OOW}) holds in $\mathrm{End}(V \otimes V \otimes V \otimes F_+ \otimes F_+ \otimes F)$. 

Here are some examples of the $RLLL$ relation (\ref{OOW}): 
\begin{align}
& (q^{i+j}-q^{a+b}) R_{i,j,k}^{a,b,c} = 0, \quad q^j R_{i,j,k-1}^{a,b,c} = q^b R_{i,j,k}^{a,b,c+1}, \\
& (\mu_2 q^{b+c} - \mu_1 q^{i+k})R_{i,j,k}^{a,b,c} =  (1-q^{2i}) t_3 w_3R_{i-1,j+1,k-1}^{a,b,c}, \\
& (\mu_2^{-1} q^{j+k}  - \mu_1^{-1} q^{a+c}) R_{i,j,k}^{a,b,c} = (1-q^{2a+2}) t_3 R_{i,j,k}^{a+1,b-1,c+1}, \\
& s_3 R_{i,j,k}^{a,b-1,c} -q^k R_{i+1,j,k}^{a,b,c} + q^{i+1} \mu_1^{-1} t_3 w_3 R_{i,j+1,k-1}^{a,b,c} = 0,
\label{84}\\
&q^k\mu_1\mu_2t_3R^{a,b,c}_{i+1,j,k-1}
+q^i\mu_2r_3s_3R^{a,b,c}_{i,j+1,k}
-q^{b+k}\mu_1s_3R^{a-1,b,c}_{i,j,k}
-q^{2+i}\mu_2t_3^2w_3R^{a,b,c}_{i,j+1,k-2}=0.
\end{align}
As these examples indicate, every recursion relation consists of those 
$R_{i,j,k}^{a,b,c}$ having the same parity of $a-c+j+k$. 

For $a, b, c, i, j, k, d \in \mathbb{Z}$, set
\begin{align}
\begin{split}
R(d)_{i,j,k}^{a,b,c} &=  \theta(e \in \Z)\theta(\min(i,j)\ge 0)\delta_{i+j}^{a+b}\,
s_3^i (\mu_2t_3)^{-a} \left(\frac{\mu_2 s_3}{t_3w_3}\right)^j\left(\frac{t_3^2w_3}{r_3s_3}\right)^e
q^{cj-bk}\frac{(q^{2+2e-2j};q^2)_j(q^{2a+2};q^2)_{i-a}}{(q^2;q^2)_f(q^{2a-2e};q^2)_{e-a}},
\end{split}
\label{roow}\\
e&=\frac{1}{2}(a-c+j+k+d),\quad f = \frac{1}{2}(b+c+i-k-d).
\label{ef1}
\end{align}
For the convenience of the proof of Theorem \ref{th:ooz}, 
we have defined $R(d)_{i,j,k}^{a,b,c}$ enlarging the range of the 
indices $a,b,i,j$ from $\Z_{\ge 0}$ to $\Z$.
We note that 
$R(d)^{a,b,c}_{i,j,k}= q^{bd}R(0)^{a,b,c}_{i,j,k+d}$ and 
$\theta(e \in \Z)\delta_{i+j}^{a+b} = \theta(f \in \Z)\delta_{i+j}^{a+b}$
since $e+f=i+j\in \Z$ holds when $a+b=i+j$.
The combinations $e$ and $f$ can be either positive or negative.

\begin{lemma}\label{le:ooz}
\begin{align}\label{bc4}
R(d)^{a,b,c}_{i,j,k} = 0 \quad \text{if }\;  \min(a,b,i,j) <0.
\end{align}
\end{lemma}
\begin{proof}
The assertion is obvious if $\min(i,j)<0$.
Thus we are to show that $\min(a,b) <0$ leads to $R(d)^{a,b,c}_{i,j,k} = 0$ 
assuming that  $\min(i,j) \ge 0$.
Suppose $a<0$. Then 
(\ref{roow}) indeed vanishes due to 
$(q^{2a+2};q^2)_{i-a} = (q^{2a+2};q^2)_\infty/(q^{2i+2};q^2)_\infty =0$.
Suppose $b<0$. We may further concentrate on the non-trivial case $e\ge a$ 
since otherwise $1/(q^{2a-2e};q^2)_{e-a}=0$.
 Then $1/(q^2;q^2)_f=0$ because of $f=  i+j -e =(a-e) +b <0$. 
\end{proof}

When $a,b,i,j\ge 0$, 
$R(d)^{a,b,c}_{i,j,k}$ is divergence-free and $R(d)^{a,b,c}_{i,j,k}=0$ 
unless $e \ge \max(a,j)$ and $f \ge 0$.
From these conditions it follows that
\begin{align}\label{oozcg}
R(d)^{a,b,c}_{i,j,k}=0 \quad \text{unless}\; 
|b-i| \le k-c+d \le b+i.
\end{align}

\begin{theorem}\label{th:ooz}
The $RLLL$ relation (\ref{OOW}) has a non-trivial solution if and only if 
$d:= \log_q\bigl(\frac{\mu_1}{\mu_2}\bigr) \in \Z$.
Up to overall normalization it is given by 
$R^{a,b,c}_{i,j,k} = R(d)^{a,b,c}_{i,j,k}$ specified by (\ref{roow}) and (\ref{ef1}).
\end{theorem}
\begin{proof}
The only if part of the first claim can be shown by an argument similar to Proposition \ref{pr:d12}.
To show the rest, one first checks that 
the formula (\ref{roow}) satisfies the recursion relation when $a,b,c,i,j,k$ are generic, i.e., 
when $\theta(\min(i,j)\ge 0)=1$.
This can be done easily since (\ref{roow}) is factorized.
The remaining task is to verify the boundary condition (\ref{bc4}) 
to assure that the contribution from the ``unwanted terms" to the recursion relation is zero.
This has been guaranteed by Lemma \ref{le:ooz}.
For example in (\ref{84}) at $b=0$, i.e.,
\begin{align}
s_3 R_{i,j,k}^{a,-1,c} -q^k R_{i+1,j,k}^{a,0,c} 
+ q^{i+1} \mu_1^{-1} t_3 w_3 R_{i,j+1,k-1}^{a,0,c} = 0,
\end{align}
the first term is unwanted.
 \end{proof}

From (\ref{roow}) and (\ref{oozcg}), $R$ is locally finite.
From (\ref{Linv}), its inverse is given by 
\begin{align}\label{oozinv}
R^{-1} = 
R\left.\right|_{\mu_i \rightarrow \mu^{-1}_i \,(i=1,2), \,r_3 \leftrightarrow s_3, 
t_3 \rightarrow t_3w_3,w_3\rightarrow w^{-1}_3},
\end{align}
where the normalization has been  
deduced from $R^{a,b,c}_{0,0,0}(d)=\delta^a_0\delta^b_0\delta^c_d$
and $R^{a,b,c}_{0,0,d}(-d) = \delta^a_0\delta^b_0\delta^c_0$.

\subsection{ZOO type}\label{ss:wwo}

We consider the $RLLL$ relation
\begin{equation}\label{WOO}
R_{456} L_{236}^O L_{135}^O L_{124}^Z 
= L_{124}^Z L_{135}^O L_{236}^OR_{456}, 
\end{equation}
where $L_{135}^O$ and $L_{236}^O$ 
are given by (\ref{LO}) with $\mu = \mu_2$ and $\mu_3$, respectively, 
and $L_{124}^Z$ is given by (\ref{LW}) with $(r, s, t, w) = (r_1, s_1, t_1, w_1)$.
In this case, $R \in \mathrm{End}(F \otimes F_+ \otimes F_+)$ and the sum (\ref{rdef}) 
extends over $a \in \mathbb{Z}$ and $b, c \in \mathbb{Z}_{\geq 0}$. 
The equality (\ref{WOO}) holds in $\mathrm{End}(V \otimes V \otimes V \otimes F \otimes F_+ \otimes F_+)$. 

Here are some examples of the $RLLL$ relation (\ref{WOO}): 
\begin{align}
& (q^{j+k}-q^{b+c}) R_{i,j,k}^{a,b,c} = 0, \quad q^j R_{i-1,j,k}^{a,b,c} = q^b R_{i,j,k}^{a+1,b,c}, \\
& (\mu_2^{-1} q^{a+b} - \mu_3^{-1} q^{i+k})R_{i,j,k}^{a,b,c} =(1-q^{2k})  t_1 R_{i-1,j+1,k-1}^{a,b,c}, \\
& (\mu_2 q^{i+j} - \mu_1q^{a+c}) R_{i,j,k}^{a,b,c} =  (1-q^{2c+2}) t_1 w_1 R_{i,j,k}^{a+1,b-1,c+1}, \\
& s_1 R_{i,j,k}^{a,b-1,c} -q^i R_{i,j,k+1}^{a,b,c} + q^{k+1} \mu_3 t_1 R_{i-1,j+1,k}^{a,b,c} = 0,
\\
&q^{b+i}\mu_2s_1R^{a,b,c-1}_{i,j,k} 
+ q^{2+k}\mu_3t_1^2w_1R^{a,b,c}_{i-2,j+1,k}
-q^it_1w_1R^{a,b,c}_{i-1,j,k+1}-q^k\mu_3r_1s_1R^{a,b,c}_{i,j+1,k}=0.
\end{align}
As these examples indicate, every recursion relation consists of those 
$R_{i,j,k}^{a,b,c}$ having the same parity of $-a+c+i+j$. 
The boundary condition is given by
\begin{align}\label{bc5}
R^{a,b,c}_{i,j,k} = 0 \quad \text{if }\;  \min(b,c,j,k) <0.
\end{align}

For $a,b,c,i,j,k,d \in\Z$, set
\begin{align}
\begin{split}
R(d)_{i,j,k}^{a,b,c} &=
\theta(e \in \Z)\theta(\min(j,k)\ge 0)\delta_{j+k}^{b+c} \,s_1^k 
\left(\frac{\mu_2}{t_1w_1}\right)^{c} \left(\frac{s_1}{\mu_2t_1}\right)^j
\left(\frac{t_1^2w_1}{r_1s_1}\right)^e q^{aj-bi}
\frac{(q^{2+2e-2j};q^2)_j(q^{2+2c};q^2)_{k-c}}{(q^2;q^2)_f(q^{2c-2e};q^2)_{e-c}},
\end{split}
\label{rwoo}\\
e &= \frac{1}{2}(-a+c+i+j-d),\quad f = \frac{1}{2}(a+b-i+k+d).
\label{ef2}
\end{align}
We note that 
$R(d)^{a,b,c}_{i,j,k} = q^{-bd}R(0)^{a,b,c}_{i-d,j,k}$ and 
$\theta(e \in \Z)\delta_{j+k}^{b+c} = \theta(f \in \Z)\delta_{j+k}^{b+c}$
since $e+f=j+k\in \Z$ when $b+c=j+k$.
The combinations $e$ and $f$ can be either positive or negative.
From $b,c,j,k \ge 0$ and the definition 
(\ref{qfac}), $R^{abc}_{ijk}(d)$ is divergence-free and $R(d)^{a,b,c}_{i,j,k}=0$ 
unless $e \ge \max(c,j)$ and $f \ge 0$.
From these conditions it follows that
\begin{align}\label{zoocg}
R(d)^{a,b,c}_{i,j,k}=0 \quad \text{unless}\; 
|b-k| \le i-a-d \le b+k.
\end{align}

\begin{theorem}\label{th:zoo}
The $RLLL$ relation (\ref{WOO}) has a non-trivial solution if and only if 
$d:= \log_q\bigl(\frac{\mu_3}{\mu_2}\bigr) \in \Z$.
Up to overall normalization it is given by 
$R^{a,b,c}_{i,j,k} = R(d)^{a,b,c}_{i,j,k}$ specified by (\ref{rwoo}) and (\ref{ef2}).
\end{theorem}

The proof is similar to Theorem \ref{th:ooz}.
From (\ref{rwoo}) and (\ref{zoocg}), $R$ is locally finite.
From (\ref{Linv}), its inverse is given by 
\begin{align}\label{zooinv}
R^{-1} = 
R\left.\right|_{\mu_i \rightarrow \mu^{-1}_i \,(i=2,3), \,r_1 \leftrightarrow s_1, 
t_1 \rightarrow t_1w_1,w_1\rightarrow w^{-1}_1},
\end{align}
where the normalization has been  
deduced from $R^{a,b,c}_{0,0,0}(d)=\delta^a_{-d}\delta^b_0\delta^c_0$
and $R^{a,b,c}_{-d,0,0}(-d) = \delta^a_0\delta^b_0\delta^c_0$.

\subsection{OZO type}\label{ss:owo}

We consider the $RLLL$ relation
\begin{align}\label{OWO}
R_{456}L^{O}_{236}L^{Z}_{135}L^{O}_{124}
=L^{O}_{124}L^{Z}_{135}L^{O}_{236}R_{456},
\end{align}
where $L^{Z}_{135}$ is given by (\ref{LW}) with 
$(r,s,t,w)= (r_2, s_2,t_2,w_2)$, and 
$L^{O}_{124}$ and $L^{O}_{236}$ are given by (\ref{LO}) 
with $\mu=\mu_1$ and $\mu_3$, respectively.
In this case, 
$R \in \mathrm{End}(F_+\otimes F \otimes F_+)$ and  the sum (\ref{rdef}) 
extends over $a,c\in \Z_{\ge 0}$ and $b\in \Z$.
The equality (\ref{OWO}) holds in  
$\mathrm{End}(V\otimes V \otimes V \otimes F_+ \otimes F \otimes F_+)$.

Here are some examples of the $RLLL$ relation (\ref{OWO}):
\begin{align}
&q^kR^{a,b,c}_{i,j-1,k}=q^cR^{a,b+1,c}_{i,j,k},
\quad
q^iR^{a,b,c}_{i,j-1,k}= q^a R^{a,b+1,c}_{i,j,k},
\\
&\mu_3 r_2R^{a-1,b,c-1}_{i,j,k} = (q^{1+a+b+c}\mu_1+q^j\mu_3)R^{a,b,c}_{i,j,k},
\\
&\mu_1 s_2R^{a,b,c}_{i+1,j,k+1} = (q^{1+i+j+k}\mu_3+q^b\mu_1)R^{a,b,c}_{i,j,k},
\\
&q^a \mu_3r_2 R^{a,b,c-1}_{i,j,k} 
+q^{b+c}(1-q^{2a+2})\mu_1R^{a+1,b,c}_{i,j,k} 
- t_2\mu_1\mu_3R^{a,b,c}_{i,j-1,k+1}=0,
\\
\begin{split}
&q^{i+j}(1-q^{2k})\mu_3r_2R^{a,b,c}_{i,j,k-1}
-q^{2+k}\mu_1t_2^2w_2R^{a,b,c}_{i+1,j-2,k}
\\
&+q^k\mu_1r_2s_2R^{a,b,c}_{i+1,j,k}-q^j(1-q^{2+2c})\mu_1\mu_3t_2R^{a,b+1,c+1}_{i,j,k}=0.
\end{split}
\end{align}
As these examples indicate, every recursion relation consists of those 
$R^{a,b,c}_{i,j,k}$ having the same parity of $a+b+c-j$.
\begin{align}\label{bc6}
R^{a,b,c}_{i,j,k} = 0 \quad \text{if }\;  \min(a,c,i,k) <0.
\end{align}

For $a,b,c,i,j,k,d \in  \Z$, set 
\begin{align}
R(d)^{a,b,c}_{i,j,k} & = \theta(e\in \mathbb{Z})
\theta(\min(i,k) \ge 0)\delta^{a-c}_{i-k} r_2^c (\mu_3t_2)^{-k} \left(\frac{\mu_3r_2}{t_2w_2}\right)^i \left(\frac{t_2^2w_2}{r_2s_2}\right)^e  q^{bk -cj}
\frac{(q^{2+2e-2k};q^2)_{k}}{(q^2;q^2)_f (q^{2i-2e};q^2)_{e-i}},
\label{rowo}\\
e &= \frac{1}{2}(i+j+k-b-d-1),\quad 
f = \frac{1}{2}(a+b+c-j+d+1).
\label{ef3}
\end{align}
We note that 
$R(d)^{a,b,c}_{i,j,k} =q^{-dk}R(0)^{a,b+d,j}_{i,j,k}$ and 
$\theta(e\in \Z)\delta^{a-c}_{i-k} = \theta(f\in \Z)\delta^{a-c}_{i-k}$
since $e+f = c+i$ when $a-c=i-k$.
The combinations $e$ and $f$ can be either positive or negative.
From $a,c,i,k \ge 0$ and the definition 
(\ref{qfac}), $R(d)^{a,b,c}_{i,j,k}$ is divergence-free and $R(d)^{a,b,c}_{i,j,k}=0$ 
unless $e \ge \max(i,k)$ and $f \ge 0$.
From these conditions it follows that
\begin{align}\label{ozocg}
R(d)^{a,b,c}_{i,j,k} = 0 \quad \text{unless}\; |a-c| \le j-b-d-1 \le a+c.
\end{align}

\begin{theorem}\label{th:owo}
The $RLLL$ relation (\ref{OWO}) has a non-trivial solution if and only if 
$d:= \log_q\bigl(-\frac{\mu_1}{\mu_3}\bigr) \in \Z$.
Up to overall normalization it is given by 
$R^{a,b,c}_{i,j,k} = R(d)^{a,b,c}_{i,j,k}$
specified by (\ref{rowo}) and (\ref{ef3}).
\end{theorem}

The proof is similar to Theorem \ref{th:ooz}.
$R$ is not locally finite.
 
\subsection{OOO type}\label{ss:ooo}
We consider the $RLLL$ relation
\begin{align}\label{OOO}
R_{456}L^{O}_{236}L^{O}_{135}L^{O}_{124}
=L^{O}_{124}L^{O}_{135}L^{O}_{236}R_{456},
\end{align}
where $L^{O}_{124}, L^{O}_{135}, L^{O}_{236}$ are 
given by (\ref{LO}) with $\mu=\mu_1, \mu_2, \mu_3$.
In this case, 
$R \in \mathrm{End}(F_+\otimes F_+ \otimes F_+)$ and  the sum (\ref{rdef}) 
extends over $a,b,c\in \Z_{\ge 0}$.
The equality (\ref{OOO}) holds in  
$\mathrm{End}(V\otimes V \otimes V \otimes F_+ \otimes F_+ \otimes F_+)$.
The problem of finding the solution to (\ref{OOO}) was studied in \cite{BS06,BMS08}.
The result has been shown \cite[eq.(2.29)]{KO12} to coincide with the 
intertwiner of the quantized coordinate ring $A_q(sl_3)$ that had been obtained earlier in \cite{KV94}.
See also the explanation in Section \ref{ss:OOOint}.
For $\mu_i$'s general, the following formula is valid (cf. \cite[eq.(3.85)]{Ku22}):
 \begin{align}\label{hp2}
 R^{a,b,c}_{i,j,k} &= \delta^{a+b}_{i+j}\delta^{b+c}_{j+k}
 \left(\frac{\mu_3}{\mu_2}\right)^i 
  \left(-\frac{\mu_1}{\mu_3}\right)^b
   \left(\frac{\mu_2}{\mu_1}\right)^k
\! q^{ik + b(k-i+1)}\binom{a+b}{a}_{q^2}
{}_2\phi_1\left({q^{-2b},q^{-2i} \atop q^{-2a-2b}}; q^2,q^{-2c}\right).
\end{align}
$R$ is obviously locally finite.
From (\ref{Linv}) and \cite[eq.(3.60)]{Ku22}, its inverse is given by
\begin{align}\label{oooinv}
R^{-1} = 
R\left.\right|_{\mu_i \rightarrow \mu^{-1}_i\,(i=1,2,3)}.
\end{align}

\begin{remark}\label{re:cro}
Let $\mathscr{R}^{a,b,c}_{i,j,k} = R^{a,b,c}_{i,j,k}\left.\right|_{\mu_i=1(i=1,2,3)}$ be the parameter-free
3D $R$ of type OOO.
It satisfies the tetrahedron equation (\ref{rrrr}). See Section \ref{ss:OOOint}.
It is known (\cite[Prop.24]{KO12}, \cite[eq.(3.63)]{Ku22})  that 
$\mathscr{R}^{a,b,c}_{i,j,k}$ is a {\em polynomial} in $q$ with integer coefficients satisfying
\begin{align}\label{cro}
\mathscr{R}^{a,b,c}_{i,j,k} = \frac{(q^2;q^2)_i(q^2;q^2)_j(q^2;q^2)_k}{(q^2;q^2)_a(q^2;q^2)_b(q^2;q^2)_c}
\mathscr{R}_{a,b,c}^{i,j,k}. 
\end{align}
When $q$ is a primitive root of unity of odd degree $N\ge 3$, it follows that $\mathscr{R}^{a,b,c}_{i,j,k}=0$ 
if $\max(i,j,k)\ge N$ and $\max(a,b,c)<N$.
It implies that the subspace $\bigoplus_{i,j,k \ge 0, \max(i,j,k) \ge N} \C |i\rangle \otimes 
 |j\rangle \otimes  |k\rangle  \subset F_+ \otimes F_+ \otimes F_+$ is invariant
 under $\mathscr{R}$.
 This fact was originally shown in \cite[Th.2.2.1b]{KS93} for each tensor component by 
 resorting to the recursion relations,  where an important consequence on the quotient was also pointed out in 
 Proposition 2.3.2 therein.
 The above proof based on (\ref{cro}) is an illuminating simplification 
 and has a natural generalization to the quantized coordinate rings 
 of other types \cite[eqs. (5.75), (8.39)]{Ku22}.
\end{remark}

\section{Solutions of $RLLL$ relation for $L = L^Z$ and $L^X$}\label{sec:ZX}

In this section we deal with the $RLLL$ relations which contain 
$L^Z$ (\ref{LW}) and $L^X$ (\ref{LX}).
As mentioned after (\ref{L3}), 
the parameters $r, s, t, w$ are assumed to be generic, hence  
the boundary conditions like 
(\ref{bc1}), (\ref{bc2}), (\ref{bc3}), (\ref{bc4}), (\ref{bc5}) and (\ref{bc6})  need not be considered.
We shall only treat the types XXZ, ZXX and XZX,  
and leave ZZX, XZZ, ZXZ and XXX cases for future study as they are considerably more complicated. 
Throughout the section, $R \in \mathrm{End}(F \otimes F \otimes F)$ with 
the sum (\ref{rdef}) 
extending over $a, b, c \in \Z$, and the 
$RLLL$ relation holds in $\mathrm{End}(V \otimes V \otimes V \otimes F \otimes F \otimes F)$.

\subsection{XXZ type}

We consider the $RLLL$ relation
\begin{align}\label{XXZ}
R_{456} L^Z_{236} L^X_{135} L^X_{124}
= L^X_{124} L^X_{135} L^Z_{236} R_{456},
\end{align}
where $L^X_{124}$ and $L^X_{135}$ are given by (\ref{LX}) with $(r, s, t, w) = (r_1, s_1, t_1, w_1)$ and $(r_2, s_2, t_2, w_2)$, respectively, and $L^Z_{236}$ is given by (\ref{LW}) with $(r, s, t, w) = (r_3, s_3, t_3, w_3)$.

Here are some examples of the $RLLL$ relation (\ref{qybe}), 
which are natural extensions of those for the OOZ type:  
\begin{align}
& (q^{i+j}-q^{a+b}) R_{i,j,k}^{a,b,c} = 0, \quad q^j R_{i,j,k-1}^{a,b,c} = q^b R_{i,j,k}^{a,b,c+1}, \\
& (s_1 t_2 w_2 q^{b+c} - s_2 t_1 w_1 q^{i+k})R_{i,j,k}^{a,b,c} 
= t_3 w_3 (r_1 s_1 - t_1^2  w_1 q^{2i}) R_{i-1,j+1,k-1}^{a,b,c}, \\
& (r_1 t_2 q^{j+k} - r_2 t_1 q^{a+c}) R_{i,j,k}^{a,b,c} = t_3 (r_1 s_1 - t_1^2 w_1q^{2a+2}) R_{i,j,k}^{a+1,b-1,c+1}, \\
& s_1 s_3 R_{i,j,k}^{a,b-1,c} -  q^k s_2 R_{i+1,j,k}^{a,b,c} + q^{i+1} t_1 t_3 w_3 R_{i,j+1,k-1}^{a,b,c} = 0,
\\
&q^bs_3t_2R^{a-1,b,c}_{i,j,k} + q^{2+i-k}t_1t_3^2w_3R^{a,b,c}_{i,j+1,k-2}
-q^{i-k}r_3s_3t_1R^{a,b,c}_{i,j+1,k}-s_2t_3R^{a,b,c}_{i+1,j,k-1}=0.
\end{align}

For $a, b, c, i, j, k \in \mathbb{Z}$, set 
\begin{align}
\begin{split}
R^{a,b,c}_{i,j,k} & = \delta_{i+j}^{a+b} \left(\frac{s_1 s_3}{s_2}\right)^i \left(\frac{s_1 t_3}{t_2}\right)^{-a} \left(\frac{s_1 s_3 t_2 w_2}{r_1 s_2 t_3 w_3}\right)^j \left(\frac{r_2 s_2}{t_2^2 w_2}\frac{t_3^2 w_3}{r_3 s_3}\right)^{g} q^{cj-bk} \\
& \quad \times \frac{(q^{b+c+i-k+2} \frac{t_1 t_2 w_2}{r_1 s_2}; q^2)_{g-a-b} (q^{h+2}\frac{t_1 w_1 t_2}{r_2 s_1}; q^2)_{g} (q^{2a+2} \frac{t_1^2 w_1}{r_1 s_1}; q^2)_{i-a}}{(q^{-b+c+i-k} \frac{r_2 t_1}{r_1 t_2}; q^2)_{g-a} (q^{h + 2}\frac{s_2 t_1 w_1}{s_1 t_2 w_2}; q^2)_{g-j}} R^{0,0,0}_{0,0,h}, 
\end{split}
\label{rxxz}\\
2g + h & = a-c+j+k\ (g \in \Z,\ h = 0, 1),
\label{gh1}
\end{align}
where $R^{0,0,0}_{0,0,0}$ and $R^{0,0,0}_{0,0,1}$ can be taken arbitrarily. 

\begin{theorem}\label{th:xxz}
Recursion relations derived from (\ref{XXZ}) consists of only those $R^{a,b,c}_{i,j,k}$' s having the same parity of $a-c+j+k$.
Each subsystem specified by $h$ admits a unique solution up to normalization, which is given by 
(\ref{rxxz})--(\ref{gh1}).
\end{theorem}

\begin{proof}
The former assertion on the parity can be verified directly.
Solving a partial set of recursion relations already leads to (\ref{rxxz})--(\ref{gh1}),
proving the uniqueness.
Then it is straightforward to check that it actually satisfies all the remaining recursion relations.
\end{proof}

$R$ is not locally finite.

Let us compare the 3D $R$ (\ref{rxxz}) for XXZ with (\ref{roow}) for OOZ. 
To fit $L^X$ in (\ref{XXZ}) to $L^O$, we specialize the parameters as
\begin{equation}\label{spe}
r_i = s_i = 1, \quad t_i = \mu_i^{-1}, \quad t_i w_i = \mu_i 
\end{equation}
for $i = 1, 2$. Then (\ref{rxxz}) becomes
\begin{align}
\begin{split}
R^{a,b,c}_{i,j,k} & = \delta_{i+j}^{a+b}\,
s_3^i (\mu_2t_3)^{-a} \left(\frac{\mu_2 s_3}{t_3w_3}\right)^j \left(\frac{t_3^2 w_3}{r_3 s_3}\right)^g q^{cj-bk} \frac{(q^{-b-c+i+k+2}\frac{\mu_1}{\mu_2}; q^2)_{j} (q^{2a+2}; q^2)_{i-a}}{(q^{-b+c+i-k} \frac{\mu_2}{\mu_1}; q^2)_{b+1}} \left(1 - q^{-h} \frac{\mu_2}{\mu_1}\right) R^{0,0,0}_{0,0,h}. 
\end{split}
\end{align}
Using the notation $e, f$ in (\ref{ef1}), where $g = e - \frac{1}{2}(h+d)$, and assuming $d \in \mathbb{Z}$ is 
so chosen that $e \in \mathbb{Z}$, this is rewritten as 
\begin{align}
\begin{split}
R^{a,b,c}_{i,j,k} & = \delta_{i+j}^{a+b}\,
s_3^i (\mu_2t_3)^{-a} \left(\frac{\mu_2 s_3}{t_3w_3}\right)^j\left(\frac{t_3^2w_3}{r_3s_3}\right)^e
q^{cj-bk}\frac{(q^{2+2e-2j-d}\frac{\mu_1}{\mu_2};q^2)_j(q^{2a+2};q^2)_{i-a}}{(q^{2+d}\frac{\mu_2}{\mu_1};q^2)_f(q^{2a-2e+d}\frac{\mu_2}{\mu_1};q^2)_{e-a}} \\
& \quad \times \frac{1 - q^{-h} \frac{\mu_2}{\mu_1}}{1 - q^d \frac{\mu_2}{\mu_1}} \left(\frac{t_3^2 w_3}{r_3 s_3}\right)^{-\frac{1}{2}(h+d)} R^{0,0,0}_{0,0,h}. 
\end{split}
\end{align}
Note that the condition $e \in \mathbb{Z}$ is equivalent to $h + d \in 2\mathbb{Z}$. 
Therefore, if $R^{0,0,0}_{0,0,h}\ (h = 0, 1)$ are taken as
\begin{equation}
\frac{1 - q^{-h} \frac{\mu_2}{\mu_1}}{1 - q^d \frac{\mu_2}{\mu_1}} \left(\frac{t_3^2 w_3}{r_3 s_3}\right)^{-\frac{1}{2}(h+d)} R^{0,0,0}_{0,0,h} \quad \rightarrow \quad \theta(h + d \in 2 \mathbb{Z}) = \theta(e \in \mathbb{Z})
\end{equation}
in the limit $\mu_1 \rightarrow \mu_2 q^d$, the 3D $R$ (\ref{roow}) 
for the OOZ case is formally reproduced.

\subsection{ZXX type}

We consider the $RLLL$ relation
\begin{align}\label{ZXX}
R_{456} L^X_{236} L^X_{135} L^Z_{124}
= L^Z_{124} L^X_{135} L^X_{236} R_{456},
\end{align}
where $L^Z_{124}$ is given by (\ref{LW}) with $(r, s, t, w) = (r_1, s_1, t_1, w_1)$, and $L^X_{135}$ and $L^X_{246}$ are given by (\ref{LX}) with $(r, s, t, w) = (r_2, s_2, t_2, w_2)$ and $(r_3, s_3, t_3, w_3)$, respectively. 

Here are some examples of the $RLLL$ relation (\ref{qybe}), which are
natural extensions of those for the ZOO type: 
\begin{align}
& (q^{j+k}-q^{b+c}) R_{i,j,k}^{a,b,c} = 0, \quad q^j R_{i-1,j,k}^{a,b,c} = q^b R_{i,j,k}^{a+1,b,c}, \\
& (s_3 t_2 q^{a+b} - s_2 t_3 q^{i+k})R_{i,j,k}^{a,b,c} = t_1 (r_3 s_3 - t_3^2 w_3  q^{2k}) R_{i-1,j+1,k-1}^{a,b,c}, \\
& (r_3 t_2 w_2 q^{i+j} - r_2 t_3 w_3 q^{a+c}) R_{i,j,k}^{a,b,c} = t_1 w_1 (r_3 s_3 -  t_3^2 w_3 q^{2c+2}) R_{i,j,k}^{a+1,b-1,c+1}, \\
& s_1 s_3 R_{i,j,k}^{a,b-1,c} -  q^i s_2 R_{i,j,k+1}^{a,b,c} + q^{k+1} t_1 t_3 w_3 R_{i-1,j+1,k}^{a,b,c} = 0,
\\
&q^bs_1t_2w_2R^{a,b,c-1}_{i,j,k}+ q^{2-i+k}t_1^2t_3w_1w_3R^{a,b,c}_{i-2,j+1,k}
-s_2t_1w_1R^{a,b,c}_{i-1,j,k+1}-q^{-i+k}r_1s_1t_3w_3R^{a,b,c}_{i,j+1,k}=0.
\end{align}
Every recursion relation consists of those $R^{a,b,c}_{i,j,k}$ having the same parity of $-a+c+i+j$. 

For $a, b, c, i, j, k \in \mathbb{Z}$, set 
\begin{align}
\begin{split}
R^{a,b,c}_{i,j,k} & = \delta_{j+k}^{b+c} \left(\frac{s_1 s_3}{s_2}\right)^k \left(\frac{s_3 t_1 w_1}{t_2 w_2}\right)^{-c} \left(\frac{s_1 s_3 t_2}{r_3 s_2 t_1}\right)^j \left(\frac{r_2 s_2}{t_2^2 w_2}\frac{t_1^2 w_1}{r_1 s_1}\right)^{g} q^{aj-bi} \\
& \quad \times \frac{(q^{a+b-i+k+2} \frac{t_2t_3 w_3}{r_3 s_2}; q^2)_{g-b-c} 
(q^{h+2}\frac{t_2 t_3 w_2}{r_2 s_3}; q^2)_{g}
 (q^{2c+2} \frac{t_3^2 w_3}{r_3 s_3}; q^2)_{k-c}}{(q^{a-b-i+k} \frac{r_2 t_3 w_3}{r_3 t_2 w_2}; q^2)_{g-c} 
 (q^{h + 2}\frac{s_2 t_3}{s_3 t_2}; q^2)_{g-j}} R^{0,0,0}_{h,0,0},
\end{split}
\label{rzxx}\\
2g + h & = -a+c+i+j\ (g \in \mathbb{Z},\ h = 0, 1), 
\label{gh2}
\end{align}
where $R^{0,0,0}_{0,0,0}$ and $R^{0,0,0}_{1,0,0}$ can be taken arbitrarily. 

\begin{theorem}\label{th:zxx}
Recursion relations derived from (\ref{ZXX})  consists of only those $R^{a,b,c}_{i,j,k}$' s having the same parity of $-a+c+i+j$.
Each subsystem specified by $h$ admits a unique solution up to normalization, which is given by 
(\ref{rzxx})--(\ref{gh2}).
\end{theorem}

The proof is similar to Theorem \ref{th:xxz}.
$R$ is not locally finite.

As the XXZ type, by specializing the parameters 
as (\ref{spe}) with $i=2,3$ 
and taking the limit $\mu_3 \rightarrow \mu_2 q^d$ with appropriate tuning of 
$R^{0,0,0}_{h,0,0}\ (h = 0, 1)$, 
one can reproduce the 3D $R$ (\ref{rwoo}) for ZOO from (\ref{rzxx}).

\subsection{XZX type}

We consider the $RLLL$ relation
\begin{align}\label{XZX}
R_{456} L^X_{236} L^Z_{135} L^X_{124}
= L^X_{124} L^Z_{135} L^X_{236} R_{456},
\end{align}
where $L^X_{124}$ and $L^X_{236}$ are given by (\ref{LX}) with $(r, s, t, w) = (r_1, s_1, t_1, w_1)$ and $(r_3, s_3, t_3, w_3)$, respectively, and $L^Z_{135}$ is given by (\ref{LW}) with $(r, s, t, w) = (r_2, s_2, t_2, w_2)$. 

Here are some examples of the $RLLL$ relation (\ref{qybe}), which 
are natural extensions of those for the OZO type:  
\begin{align}
& q^k R^{a,b,c}_{i,j-1,k}=q^c R^{a,b+1,c}_{i,j,k},\quad q^i R^{a,b,c}_{i,j-1,k}= q^a R^{a,b+1,c}_{i,j,k},\\
& r_2 R^{a-1,b,c-1}_{i,j,k} = (t_1 t_3 w_1  q^{1+a+b+c} + r_1 r_3 q^j) R^{a,b,c}_{i,j,k}, \\
&s_2 R^{a,b,c}_{i+1,j,k+1} = (t_1 t_3 w_3 q^{1+i+j+k} + s_1 s_3 q^b) R^{a,b,c}_{i,j,k}, \\
& q^a r_2 t_1 R^{a,b,c-1}_{i,j,k} + t_3 (r_1 s_1 -  t_1^2 w_1 q^{2a+2}) q^{b+c} 
R^{a+1,b,c}_{i,j,k} - r_1 t_2 R^{a,b,c}_{i,j-1,k+1}=0,
\\
\begin{split}
&q^{-b+c}r_2s_2t_3w_3R^{a-1,b,c}_{i,j,k}
- q^{-b+c}t_2^2t_3w_2w_3R^{a-1,b+2,c}_{i,j,k}
\\
&\;\;-s_1t_2w_2(r_3s_3-t_3^2w_3q^{2k})R^{a,b,c}_{i,j-1,k-1}
+q^as_2t_1w_1(r_3s_3-t_3^2w_3q^{2+2c})R^{a,b,c+1}_{i,j,k}=0.
\end{split}
\end{align}

For $a, b, c, i, j, k \in \mathbb{Z}$, set 
\begin{align}
\begin{split}
R^{a,b,c}_{i,j,k} & = \delta^{a-c}_{i-k} \left(\frac{r_2}{r_1 r_3}\right)^c 
\left(\frac{s_1 t_3}{t_2}\right)^k \left(\frac{r_2 t_3 w_3}{r_3 t_2 w_2}\right)^i 
\left(\frac{r_3 s_3}{t_3^2 w_3}\frac{t_2^2 w_2}{r_2 s_2}\right)^{g} q^{bk-cj} \\
& \quad \times \frac{(-q^{h+1}\frac{t_1 t_3 w_3}{s_1 s_3}; q^2)_{g}}{(-q^{h+1}
\frac{r_3 t_1}{s_1 t_3}; q^2)_{g-k}} 
\frac{(-q^{-h+1} \frac{s_3 t_1 w_1}{r_1 t_3 w_3}; q^2)_{i-g}}{(-q^{-h+3}
 \frac{t_1 t_3w_1}{r_1 r_3}; q^2)_{c+i-g}} R^{0,0,0}_{0,h,0}
\end{split}
\label{rxzx}\\
2g+h & = -b+i+j+k\ (g \in \mathbb{Z},\ h = 0, 1),
\label{gh3}
\end{align}
where $R^{0,0,0}_{0,0,0}$ and $R^{0,0,0}_{0,1,0}$ can be taken arbitrarily. 

\begin{theorem}\label{th:xzx}
Recursion relations derived from (\ref{XZX}) 
consists of only those $R^{a,b,c}_{i,j,k}$' s having the same parity of $-b+i+j+k$.
Each subsystem specified by $h$ admits a unique solution up to normalization, which is given by 
(\ref{rxzx})--(\ref{gh3}).
\end{theorem}

The proof is similar to Theorem \ref{th:xxz}.
$R$ is not locally finite.

Let us compare the 3D $R$ (\ref{rxzx}) for XZX with (\ref{rowo}) for OZO. 
To fit $L^X$ in (\ref{XZX}) to $L^O$, we specialize the parameters as
(\ref{spe}) with $i=1,3$.
Then (\ref{rxzx}) becomes
\begin{align}
\begin{split}
R^{a,b,c}_{i,j,k} & = \delta_{i-k}^{a-c} r_2^c(\mu_3t_2)^{-k}
\left(\frac{\mu_3r_2}{t_2w_2}\right)^i
\left(\frac{r_2 s_2}{t_2^2 w_2}\right)^{g} q^{bk-cj} \frac{(-q^{-b+i+j-k+1}\frac{\mu_3}{\mu_1}; q^2)_k}{(-q^{a+b-c-j+1}\frac{\mu_1}{\mu_3}; q^2)_{c+1}} \left(1+q^{-h+1} \frac{\mu_1}{\mu_3}\right)R^{0,0,0}_{0,h,0}. 
\end{split}
\end{align}
Using the notation $e, f$ in (\ref{ef3}), where $g = e - \frac{1}{2}(h-d-1)$, 
and assuming $d \in \mathbb{Z}$ is so chosen that $e \in \mathbb{Z}$, 
this is rewritten as 
\begin{align}
\begin{split}
R^{a,b,c}_{i,j,k} & = \delta_{i+j}^{a+b} r_2^c(\mu_3t_2)^{-k}
\left(\frac{\mu_3r_2}{t_2w_2}\right)^i \left(\frac{t_2^2w_2}{r_2s_2}\right)^e q^{bk-cj} \frac{(-q^{2e-2k+d+2}\frac{\mu_3}{\mu_1}; q^2)_k}{(-q^{-d+2} \frac{\mu_1}{\mu_3} ;q^2)_f (-q^{-2e+2i-d} \frac{\mu_1}{\mu_3} ;q^2)_{e-i}} \\
& \quad \times \frac{1+q^{-h+1} \frac{\mu_1}{\mu_3}}{1+q^{-d} \frac{\mu_1}{\mu_3}} \left(\frac{t_2^2 w_2}{r_2 s_2}\right)^{-\frac{1}{2}(h-d-1)} R^{0,0,0}_{0,h,0}. 
\end{split}
\end{align}
Note that the condition $e \in \mathbb{Z}$ 
is equivalent to $h - d - 1 \in 2\mathbb{Z}$. 
Therefore, if $R^{0,0,0}_{0,h,0}\ (h = 0, 1)$ are taken as
\begin{equation}
\frac{1+q^{-h+1} \frac{\mu_1}{\mu_3}}{1+q^{-d} \frac{\mu_1}{\mu_3}} \left(\frac{t_2^2 w_2}{r_2 s_2}\right)^{-\frac{1}{2}(h-d-1)} R^{0,0,0}_{0,h,0} \quad \rightarrow \quad \theta(h - d - 1 \in 2 \mathbb{Z}) = \theta(e \in \mathbb{Z})
\end{equation}
in the limit $\mu_1 \rightarrow -\mu_3 q^d$, 
the 3D $R$ (\ref{rowo}) for the OZO case is formally reproduced.

\section{Relation to the representation theory of the quantized coordinate ring.}\label{sec:rep}

\subsection{Quantized coordinate ring $A_q(sl_3)$}
The algebra $A_q(sl_3)$ is a Hopf algebra dual to 
the quantized universal enveloping algebra $U_q(sl_3)$.
See for example \cite{Dri86, RTF90, Kas93, CP94, KS98} and the references therein.
It is generated by $t_{ij}\,(1 \le i,j \le 3)$ with the relations
\begin{align}
&[t_{ik}, t_{jl}]=\begin{cases}0 &(i<j, k>l),\\
(q-q^{-1})t_{jk}t_{il} & (i<j, k<l),
\end{cases}
\label{Aqr1}\\
&t_{ik}t_{jk} = q t_{jk}t_{ik}\; (i<j),\quad 
t_{ki}t_{kj} = q t_{kj}t_{ki}\; (i<j),
\label{Aqr2}\\
&\sum_{\sigma \in \mathfrak{S}_3}(-q)^{l(\sigma)}
t_{1\sigma_1}t_{2\sigma_2} t_{3\sigma_3} = 1,
\label{Aqr3}
\end{align} 
where $\mathfrak{S}_3$ denotes the symmetric group of degree 3 and 
$l(\sigma)$ is the length of the permutation $\sigma$.
The coproduct 
$\Delta: A_q(sl_3) \rightarrow A_q(sl_3)^{\otimes N}$ is given by the matrix product form
$\Delta t_{ij} = \sum_{1 \le i_2,\ldots, i_N\le 3} t_{i  i_2} \otimes t_{i_2 i_3} \otimes \cdots \otimes t_{i_N j}$.

The following maps define the algebra homomorphisms to the $q$-Weyl algebra (\ref{qw}):
\begin{align}
\begin{split}
\rho_1: A_q(sl_3) &\rightarrow \mathcal{W}_q,
\\
\begin{pmatrix} 
t_{11} & t_{12} & t_{13}\\
t_{21} & t_{22} & t_{23} \\
t_{31} & t_{32} & t_{33}
\end{pmatrix} 
&\mapsto
\begin{pmatrix}
Z^{-1}(u_1 - g_1 h_1  X^2) & g_1 X & 0
\\
-qh_1 X & Z & 0
\\
0 & 0 & u_1^{-1}
\end{pmatrix},
\end{split}
\\
\begin{split}
\rho_2: A_q(sl_3) &\rightarrow \mathcal{W}_q,
\\
\begin{pmatrix} 
t_{11} & t_{12} & t_{13}\\
t_{21} & t_{22} & t_{23} \\
t_{31} & t_{32} & t_{33}
\end{pmatrix} 
&\mapsto
\begin{pmatrix}
u_2^{-1} & 0 & 0
\\
0 & Z^{-1}(u_2 - g_2 h_2  X^2) & g_2 X 
\\
0 & -qh_2 X & Z
\end{pmatrix}.
\end{split}
\end{align}
Here, $u_i, g_i, h_i$ are arbitrary parameters.

We let 
$\rho_{Z,i} = \pi_Z \circ \rho_i$ 
and 
$\rho_{X,i} = \pi_X \circ \rho_i$
denote the representations
$A_q(sl_3) \rightarrow \mathrm{End}(F)$
obtained by the compositions with 
$\pi_Z$ and $\pi_X$ in (\ref{piz}) and (\ref{pix}).

\subsection{3D $R$ of type OOO as an intertwiner of $\rho_{O,i}$.}\label{ss:OOOint}

From the remark after (\ref{lwinv}),
one can restrict $\rho_{X,i}$ 
with $(u_i, g_i, h_i) = 
(1,\mu_i, \mu_i^{-1})$ from $\mathrm{End}(F)$ to $\mathrm{End}(F_+)$.
The resulting representation will be denoted by 
$\rho_{O,i}: A_q(sl_3) \rightarrow \mathrm{End}(F_+)$.

The representation $\rho_{O,i}$ is irreducible and well-studied \cite{KS98}.
In fact, the isomorphism of the tensor product representations
$\rho_{O,1} \otimes \rho_{O,2} \otimes \rho_{O,1}
\simeq \rho_{O,2} \otimes \rho_{O,1} \otimes \rho_{O,2}$ 
is valid, and they turn out to be irreducible.
Let  $\Phi \in \mathrm{End}(F^{\otimes 3}_+)$ be the intertwiner, i.e., 
the unique solution to the intertwining relation 
$\Phi \circ (\rho_{O,1} \otimes \rho_{O,2} \otimes \rho_{O,1}) = 
(\rho_{O,2} \otimes \rho_{O,1} \otimes \rho_{O,2}) \circ \Phi$
up to normalization.
Set $\mathscr{R} = \Phi\circ P$ 
where $P$ is the transposition $P(|i \rangle \otimes |j \rangle \otimes |k \rangle) = 
|k \rangle \otimes |j \rangle \otimes |i \rangle$.
We also call $\mathscr{R}$ the intertwiner.
The intertwining relation for the generator $t_{l m}$ reads as 
\begin{align}\label{ir1}
\mathscr{R}
(\rho_{O,1} \otimes \rho_{O,2} \otimes \rho_{O,1}) (\Delta\!' t_{l m})
= (\rho_{O,2} \otimes \rho_{O,1} \otimes \rho_{O,2})(\Delta t_{l m})
\mathscr{R}
\qquad (1 \le l, m \le 3),
\end{align}
where $\Delta\!' = P \circ \Delta \circ P$, hence 
$\Delta\!' t_{l m} =
\sum_{j,k}t_{k m} \otimes t_{jk} \otimes t_{l j}$.
It is known that the set of equations (\ref{ir1}) are equivalent to 
the $RLLL$ relation $(\ref{OOO})|_{\mu_3 = \mu_1}$
under the identification $\mathscr{R}=R$.
See \cite[Sec.2]{KO12} and \cite[Lem.3.22]{Ku22}.
As a result, the 3D $R$ (\ref{hp2}) with $\mu_3=\mu_1$ 
is identified with the 
intertwiner of the $A_q(sl_3)$ modules.

\subsection{3D $R$ of type ZZZ as an intertwiner of $\rho_{Z,i}$.}\label{ss:ZZZint}

Consider the  equation on $R \in \mathrm{End}(F^{\otimes 3})$ given by
\begin{align}\label{ir2}
R (\rho_{Z,1} \otimes \rho_{Z,2} \otimes \rho_{Z,1}) (\Delta\!' t_{l m})
= (\rho_{Z,2} \otimes \rho_{Z,1} \otimes \rho_{Z,2})(\Delta t_{l m}) R
\qquad (1 \le l, m \le 3),
\end{align}
which includes the parameters 
$u_i, g_i, h_i (i = 1,2)$.
On the other hand, recall that the $RLLL$ relation (\ref{ZZZ}) of ZZZ type 
contains 18 equations depending on $r_\alpha, s_\alpha, t_\alpha, w_\alpha (\alpha = 1,2,3)$.
Now we state a result analogous to the OOO case in the previous subsection.

\begin{proposition}\label{pr:z}
The intertwining relation (\ref{ir2}) and the $RLLL$ relation (\ref{ZZZ}) are equivalent 
provided that the parameters in the former obey the constraint 
$u_1 = u_2( =:\!u)$ and  $g_1h_1 = g_2h_2(=:\!p)$, and those in the latter satisfy
\begin{align}\label{rrtt}
\frac{r_1}{t_1} = \frac{r_2}{t_2},\quad
\frac{s_2}{t_2} = \frac{s_3}{t_3},
\quad
\frac{r_2}{r_1r_3}=u,
\quad
\frac{s_1s_3}{s_2} = u^2,
\quad
\frac{t_1^2w_1}{r_1s_1}= 
\frac{t_2^2w_2}{r_2s_2}= 
\frac{t_3^2w_3}{r_3s_3}= \frac{p}{u}.
\end{align}
\end{proposition}
\begin{proof}
Set 
$\mathcal{L}^{abc}_{ijk} = \sum_{\alpha, \beta, \gamma}
\mathcal{L}^{\alpha \beta}_{ij} 
\otimes \mathcal{L}^{a \gamma}_{\alpha k} \otimes \mathcal{L}^{bc}_{\beta\gamma}$
and 
$\tilde{\mathcal{L}}^{abc}_{ijk} = \sum_{\alpha, \beta, \gamma}
\mathcal{L}^{ab}_{\alpha\beta} \otimes 
\mathcal{L}^{\alpha c}_{i\gamma} \otimes \mathcal{L}^{\beta\gamma}_{jk}$
so that (\ref{qybe}) reads as $R \mathcal{L}^{abc}_{ijk} = \tilde{\mathcal{L}}^{abc}_{ijk}R$.
Then one can directly check that the relations (\ref{rrtt}) validate the equalities
\begin{align}
(\rho_{Z,1} \otimes \rho_{Z,2} \otimes \rho_{Z,1}) (\Delta\!' t_{l m})
&= A_{l m} \mathcal{L}^{abc}_{ijk} 
= B_{lm}
\left(\tilde{\mathcal{L}}^{a'b'c'}_{i'j'k'}\left|_{r_\gamma \leftrightarrow s_\gamma,
t_\gamma \rightarrow 
t_\gamma w_\gamma, w_\gamma \rightarrow w_\gamma^{-1}}  \right.\right),
\label{AB1}
\\
(\rho_{Z,2} \otimes \rho_{Z,1} \otimes \rho_{Z,2})(\Delta t_{l m}) 
&= A_{l m}\tilde{\mathcal{L}}^{abc}_{ijk}
=B_{l m} \left(
\mathcal{L}^{a'b'c'}_{i'j'k'}\left|_{r_\gamma \leftrightarrow s_\gamma,
t_\gamma \rightarrow t_\gamma w_\gamma, 
w_\gamma \rightarrow w_\gamma^{-1}}  \right.\right),
\label{AB2}
\end{align}
where the constants $A_{lm}, B_{lm}$ are given by
\begin{align}
(A_{lm})_{1\le l, m \le 3} &= \begin{pmatrix}
\begin{array}{ccc}
 \frac{1}{r_2^2 s_2} & \frac{p u}{h_1 r_2 s_2 t_3} & \frac{p^2}{h_1 h_2 u} \\
 -\frac{q h_1 t_3}{p r_2^2 s_2} & -\frac{u}{r_2 s_2} & \frac{p}{h_2 r_2 t_1 u} \\
 \frac{q^2 h_1 h_2}{u} & -\frac{q h_2 t_1 u}{p r_2 s_2} & \frac{1}{r_2 u} \\
\end{array}
\end{pmatrix},
\\
(B_{lm})_{1\le l, m \le 3} &= \begin{pmatrix}
\begin{array}{ccc}
 -\frac{1}{r_2^2 s_2} & \frac{p u}{h_1 r_2 s_2 t_3} & \frac{p^2}{h_1 h_2 u} \\
 -\frac{q h_1 t_3}{p r_2^2 s_2} & \frac{u}{r_2 s_2} & \frac{p}{h_2 r_2 t_1 u} \\
 \frac{q^2 h_1 h_2}{u} & -\frac{q h_2 t_1 u}{p r_2 s_2} & \frac{1}{r_2 u} \\
\end{array}
\end{pmatrix}.
\end{align}
The correspondence between the indices $l, m$ and 
$a,b,c,i,j,k, a',b',c',i',j',k'$ is specified as follows:
\begin{table}[h]
 \begin{minipage}[t]{.45\textwidth}
 \begin{center}
\begin{tabular}{c| c| c }
$l$ & $abc$ & $i'j'k'$
\\
\hline
1 & 001 & 011\\
2 & 010 & 101 \\
3 & 100 & 110\\
\end{tabular}
 \end{center}
  \end{minipage}
  \hspace{-4cm}
  \begin{minipage}[t]{.45\textwidth}
    \begin{center}
\begin{tabular}{c| c|c }
$m$ & $ijk$ & $a'b'c'$
\\
\hline
1 & 100 & 110\\
2 & 010 & 101 \\
3 & 001 & 011
\end{tabular}
    \end{center}
  \end{minipage}
\end{table}

\noindent
The relations (\ref{AB1}) and (\ref{AB2}) including $A_{lm}$ enable us to identify 
(\ref{ir2}) with 
$R \mathcal{L}^{abc}_{ijk} = \tilde{\mathcal{L}}^{abc}_{ijk}R$,
covering the case $a+b+c=i+j+k=1$ of the latter.
Let us show the other case $a'+b'+c'=i'+j'+k'=2$ of 
the $RLLL$ relation in the form
$R^{-1}\tilde{\mathcal{L}}^{a'b'c'}_{i'j'k'} = \mathcal{L}^{a'b'c'}_{i'j'k'}R^{-1}$.
Due to (\ref{zzzinv}) it is equivalent to 
$R\left(\tilde{\mathcal{L}}^{a'b'c'}_{i'j'k'}\left|_{r_\gamma \leftrightarrow s_\gamma,
t_\gamma \rightarrow 
t_\gamma w_\gamma, w_\gamma \rightarrow w_\gamma^{-1}}  \right.\right)
=\left(
\mathcal{L}^{a'b'c'}_{i'j'k'}\left|_{r_\gamma \leftrightarrow s_\gamma,
t_\gamma \rightarrow t_\gamma w_\gamma, 
w_\gamma \rightarrow w_\gamma^{-1}}  \right.\right)R$.
This equality follows from (\ref{ir2}) by
applying the relations (\ref{AB1}) and (\ref{AB2}) including $B_{lm}$.
\end{proof}

\section{Discussion}\label{sec:dis}

\subsection{Summary}
In this paper we have studied the tetrahedron equation of the form 
$R_{456}L^C_{236}L^B_{135}L^A_{124}  = L^A_{124}L^B_{135}L^C_{236}R_{456}$ 
for the three kinds of 3D $L$ operators $L^Z, L^X, L^O$ in 
(\ref{LW})--(\ref{LO}) which can be regarded as quantized six-vertex models with 
Boltzmann weights taken from the $q$-Weyl algebra $\mathcal{W}_q$ (\ref{qw}) 
or the $q$-oscillator algebra $\mathcal{O}_q$ (\ref{qoa}).
In each case the solution $R$ has been obtained explicitly whose elements are factorized or expressed in terms of 
terminating $q$-hypergeometric type series as in Table \ref{tab:kekka}.
They are new except for the OOO case.

\begin{table}[H]
\begin{center}
\begin{tabular}{c|c|c|c|c|c}
ABC & $\sharp (\mathrm{Z})$ & feature & ${\rm locally \atop \rm \underset{}{finiteness}}$
& $\sharp$(sector) &formula \\
\hline
ZZZ & 3 & factorized & no & 4 & (\ref{rwww})\\
\hline
OZZ &   & ${}_2\phi_1$ & no&   & (\ref{roww2})\\
ZZO & 2 & ${}_2\phi_1$ & no& 1 & (\ref{rwwo2})\\
ZOZ &   & ${}_3\phi_2$-like & no&  & (\ref{rwow2})\\
\hline
OOZ &   & factorized & yes &   & (\ref{roow})\\
ZOO & 1 & factorized & yes & 1 & (\ref{rwoo})\\
OZO &   & factorized & no &   & (\ref{rowo})\\
\hline
OOO & 0 & ${}_2\phi_1$ & yes & 1 & (\ref{hp2})\\
\hline
XXZ &   & factorized & no &   & (\ref{rxxz})\\
ZXX & 1 & factorized & no & 2 & (\ref{rzxx})\\
XZX &   & factorized & no &   & (\ref{rxzx})
\end{tabular}
\end{center}\caption{Type ABC of $R_{456}L^C_{236}L^B_{135}L^A_{124}  = L^A_{124}L^B_{135}L^C_{236}R_{456}$ 
and the basic feature of the solution $R=R^{ABC}$.
We observe the factorization when the number $\sharp ({\rm Z})$ of Z in ABC is odd.
$\sharp$(sector) is the dimension of the solution space 
for the recursion relations of $R^{a,b,c}_{i,j,k}$.}
\label{tab:kekka}
\end{table}

\subsection{On tetrahedron equation of the form $RRRR= RRRR$}

Let us discuss the tetrahedron equation of the form
\begin{align}\label{rrrr}
R_{456}R_{236}R_{135}R_{124}=R_{124}R_{135}R_{236}R_{456}.
\end{align}
A standard strategy for the proof is to compare the two maneuvers:
\begin{equation}\label{SL1}
\begin{split}
&R_{124}R_{135}R_{236}R_{456}
\underline{L_{\alpha \beta 6}L_{\alpha \gamma 5}L_{\beta \gamma 4}}L_{\alpha \delta 3}
L_{\beta \delta 2}L_{\gamma \delta 1}\\
&=R_{124}R_{135}R_{236}L_{\beta \gamma 4}L_{\alpha \gamma 5}
\underline{L_{\alpha \beta 6}L_{\alpha \delta 3}L_{\beta \delta 2}}L_{\gamma \delta 1}R_{456}\\
&=R_{124}R_{135}L_{\beta \gamma 4}\underline{L_{\alpha \gamma 5}L_{\beta \delta 2}}
L_{\alpha \delta 3}\underline{L_{\alpha \beta 6}L_{\gamma \delta 1}}
R_{236}R_{456}\\
&=R_{124}R_{135}L_{\beta \gamma 4}L_{\beta \delta 2}
\underline{L_{\alpha \gamma 5}L_{\alpha \delta 3}L_{\gamma \delta 1}}L_{\alpha \beta 6}
R_{236}R_{456}\\
&=R_{124}\underline{L_{\beta \gamma 4}L_{\beta \delta 2}L_{\gamma \delta 1}}
L_{\alpha \delta 3}L_{\alpha \gamma 5}L_{\alpha \beta 6}
R_{135}R_{236}R_{456}\\
&=L_{\gamma \delta 1}L_{\beta \delta 2}\underline{L_{\beta \gamma 4}L_{\alpha \delta 3}}
L_{\alpha \gamma 5}L_{\alpha \beta 6}
R_{124}R_{135}R_{236}R_{456},\\
&=L_{\gamma \delta 1}L_{\beta \delta 2}L_{\alpha \delta 3}L_{\beta \gamma 4}
L_{\alpha \gamma 5}L_{\alpha \beta 6}
R_{124}R_{135}R_{236}R_{456},
\end{split}
\end{equation}
\begin{equation}\label{SL2}
\begin{split}
&R_{456}R_{236}R_{135}R_{124}
L_{\alpha \beta 6}L_{\alpha \gamma 5}\underline{L_{\beta \gamma 4}L_{\alpha \delta 3}}
L_{\beta \delta 2}L_{\gamma \delta 1}\\
&=R_{456}R_{236}R_{135}R_{124}
L_{\alpha \beta 6}L_{\alpha \gamma 5}L_{\alpha \delta 3}
\underline{L_{\beta \gamma 4}L_{\beta \delta 2}L_{\gamma \delta 1}}\\
&=R_{456}R_{236}R_{135}
L_{\alpha \beta 6}\underline{L_{\alpha \gamma 5}L_{\alpha \delta 3}L_{\gamma \delta 1}}
L_{\beta \delta 2}L_{\beta \gamma 4}R_{124}\\
&=R_{456}R_{236}
\underline{L_{\alpha \beta 6}L_{\gamma \delta 1}}L_{\alpha \delta 3}
\underline{L_{\alpha \gamma 5}L_{\beta \delta 2}}L_{\beta \gamma 4}R_{135}R_{124}\\
&=R_{456}R_{236}
L_{\gamma \delta 1}\underline{L_{\alpha \beta 6}L_{\alpha \delta 3}L_{\beta \delta 2}}
L_{\alpha \gamma 5}L_{\beta \gamma 4}R_{135}R_{124}\\
&=R_{456}
L_{\gamma \delta 1}L_{\beta \delta 2}L_{\alpha \delta 3}
\underline{L_{\alpha \beta 6}L_{\alpha \gamma 5}L_{\beta \gamma 4}}R_{236}R_{135}R_{124}\\
&=L_{\gamma \delta 1}L_{\beta \delta 2}L_{\alpha \delta 3}
L_{\beta \gamma 4}L_{\alpha \gamma 5}L_{\alpha \beta 6}R_{456}R_{236}R_{135}R_{124}.
\end{split}
\end{equation}
The underlines indicate the components to be rewritten by the 
$RLLL=LLLR$ relation or trivial commutativity of the operators acting on distinct set of components. 
The above relations show that the composition 
$(R_{124}R_{135}R_{236}R_{456})^{-1}R_{456}R_{236}R_{135}R_{124}$
commutes with 
$L_{\alpha \beta 6}L_{\alpha \gamma 5}L_{\beta \gamma 4}L_{\alpha \delta 3}
L_{\beta \delta 2}L_{\gamma \delta 1}$.
Therefore if the action of the latter is irreducible,  Schur's lemma compels
$R_{124}R_{135}R_{236}R_{456} = (\text{scalar}) R_{456}R_{236}R_{135}R_{124}$ and 
the scalar can be fixed by considering the special case.

In this type of argument, $RLLL=LLLR$ serves as an auxiliary linear problem for $RRRR=RRRR$,
which is analogous to the quantum group symmetry ensuring the Yang-Baxter equation. 
It indeed works when all the $L$'s are $L^O$, where 
$RLLL=LLLR$ is identified with the intertwining relation of the 
quantized coordinate ring $A_q(sl_3)$. See Section \ref{ss:OOOint}.
The corresponding 3D $R$ of type OOO (\ref{hp2}) certainly satisfies the tetrahedron equation
$R_{124}R_{135}R_{236}R_{456} = R_{456}R_{236}R_{135}R_{124}$ \cite{KV94,BS06}.

The results in this paper suggest a natural generalization where  
the six $L$ operators in (\ref{SL1}) and (\ref{SL2}) are taken either as $L^Z$ or $L^O$
(resp. $L^Z$ or $L^X$) in the context of Section \ref{sec:ZO} (resp. Section \ref{sec:ZX}).
Let us exhibit them as 
$L_{\alpha \beta 6}^FL_{\alpha \gamma 5}^EL_{\beta \gamma 4}^DL_{\alpha \delta 3}^C
L_{\beta \delta 2}^BL_{\gamma \delta 1}^A$, where 
A, B, C, D, E and F assume  Z, O or X.
The corresponding generalization of (\ref{rrrr}) reads as 
\begin{align}\label{rrrr2}
R_{456}^{DEF}R_{236}^{BCF}R_{135}^{ACE}R_{124}^{ABD}
=R_{124}^{ABD}R_{135}^{ACE}R_{236}^{BCF}R_{456}^{DEF},
\end{align}
where $R_{124}^{ABD}$ for example denotes the 3D $R$ of type ABD 
acting on the tensor components 1, 2 and 4.
Let us call (\ref{rrrr2})  the $RRRR$ relation of type ABCDEF.
Its proof or disproof is an important future problem.
It has been settled only for type OOOOOO as explained in the above. 
However, the argument employed there does not persist naively when $L^Z$ is involved 
since the irreducibility no longer holds due to Remark \ref{re:w}. 
Moreover, the presence of locally non-finite 3D $R$
makes the convergence of the compositions in $RRRR=RRRR$ non-trivial.\footnote{The convergence 
also matters when one attempts 
to perform the {\em reduction} of the 3D $R$'s to the 
solutions of the Yang-Baxter equation by the trace \cite{BS06} 
and the boundary vectors \cite{KuS13,KO15}. }
In spite of such difficulties,  we have made promising observations
which are reported below.

From Table \ref{tab:kekka}, the tetrahedron equation (\ref{rrrr2})
consisting of only locally finite 3D $R$' s are of type OOOOOO and the following:
\begin{align}
\text{Type $\quad\;\;$} \quad &   
\text{$\qquad\qquad \qquad $Tetrahedron equation}
\nonumber \\
\text{ZOOOOO}:  \quad &R_{456}^{OOO}R_{236}^{OOO}R_{135}^{ZOO}R_{124}^{ZOO}
=R_{124}^{ZOO}R_{135}^{ZOO}R_{236}^{OOO}R_{456}^{OOO},
\label{z1}\\
\text{OOOZOO}:  \quad &R_{456}^{ZOO}R_{236}^{OOO}R_{135}^{OOO}R_{124}^{OOZ}
=R_{124}^{OOZ}R_{135}^{OOO}R_{236}^{OOO}R_{456}^{ZOO},
\label{z4}\\
\text{OOOOOZ}:  \quad &R_{456}^{OOZ}R_{236}^{OOZ}R_{135}^{OOO}R_{124}^{OOO}
=R_{124}^{OOO}R_{135}^{OOO}R_{236}^{OOZ}R_{456}^{OOZ},
\\
\text{ZOOOOZ}:  \quad &R_{456}^{OOZ}R_{236}^{OOZ}R_{135}^{ZOO}R_{124}^{ZOO}
=R_{124}^{ZOO}R_{135}^{ZOO}R_{236}^{OOZ}R_{456}^{OOZ}.
\label{z16}
\end{align}
In these equations,  images of any given input vector 
$|i\rangle \otimes |j\rangle \otimes |k\rangle \otimes |l\rangle \otimes |m\rangle \otimes |n\rangle$ 
by the two sides are linear combinations of finitely many bases
with finite coefficients, so one can compare them directly.

The tetrahedron equation (\ref{rrrr2})
containing {\em only one} locally non-finite 3D $R$ on each side are the following:
\begin{align}
\text{Type $\quad\;\;$} \quad &   
\text{$\qquad\qquad \qquad $Tetrahedron equation}
\nonumber \\
\text{OZOOOO}:  \quad &R_{456}^{OOO}R_{236}^{ZOO}R_{135}^{OOO}R_{124}^{OZO}
=R_{124}^{OZO}R_{135}^{OOO}R_{236}^{ZOO}R_{456}^{OOO},
\label{z2}\\
\text{OOOOZO}:  \quad &R_{456}^{OZO}R_{236}^{OOO}R_{135}^{OOZ}R_{124}^{OOO}
=R_{124}^{OOO}R_{135}^{OOZ}R_{236}^{OOO}R_{456}^{OZO},
\\
\text{ZZOOOO}:  \quad &R_{456}^{OOO}R_{236}^{ZOO}R_{135}^{ZOO}R_{124}^{ZZO}
=R_{124}^{ZZO}R_{135}^{ZOO}R_{236}^{ZOO}R_{456}^{OOO},
\\
\text{ZOOZOO}:  \quad &R_{456}^{ZOO}R_{236}^{OOO}R_{135}^{ZOO}R_{124}^{ZOZ}
=R_{124}^{ZOZ}R_{135}^{ZOO}R_{236}^{OOO}R_{456}^{ZOO},
\\
\text{OZOZOO}:  \quad &R_{456}^{ZOO}R_{236}^{ZOO}R_{135}^{OOO}R_{124}^{OZZ}
=R_{124}^{OZZ}R_{135}^{OOO}R_{236}^{ZOO}R_{456}^{ZOO},
\\
\text{OOOZZO}:  \quad &R_{456}^{ZZO}R_{236}^{OOO}R_{135}^{OOZ}R_{124}^{OOZ}
=R_{124}^{OOZ}R_{135}^{OOZ}R_{236}^{OOO}R_{456}^{ZZO},
\\
\text{OOOZOZ}:  \quad &R_{456}^{ZOZ}R_{236}^{OOZ}R_{135}^{OOO}R_{124}^{OOZ}
=R_{124}^{OOZ}R_{135}^{OOO}R_{236}^{OOZ}R_{456}^{ZOZ},
\\
\text{OOOOZZ}:  \quad &R_{456}^{OZZ}R_{236}^{OOZ}R_{135}^{OOZ}R_{124}^{OOO}
=R_{124}^{OOO}R_{135}^{OOZ}R_{236}^{OOZ}R_{456}^{OZZ},
\\
\text{ZZOZOO}:  \quad &R_{456}^{ZOO}R_{236}^{ZOO}R_{135}^{ZOO}R_{124}^{ZZZ}
=R_{124}^{ZZZ}R_{135}^{ZOO}R_{236}^{ZOO}R_{456}^{ZOO},
\label{z124}
\\
\text{OOOZZZ}:  \quad &R_{456}^{ZZZ}R_{236}^{OOZ}R_{135}^{OOZ}R_{124}^{OOZ}
=R_{124}^{OOZ}R_{135}^{OOZ}R_{236}^{OOZ}R_{456}^{ZZZ}.
\end{align}
In these equations, transition amplitudes for any {\em pair} of input and output bases 
$|i\rangle \otimes |j\rangle \otimes |k\rangle \otimes |l\rangle \otimes |m\rangle \otimes |n\rangle
\rightarrow 
|a\rangle \otimes |b\rangle \otimes |c\rangle \otimes |d\rangle \otimes |e\rangle \otimes |f\rangle$
by the two sides are finite.
Explcitly they are the two sides of 
\begin{align}\label{rrrr3}
\sum_{u,v,w,x,y,z}
R^{d,e,f}_{x,y,z}R^{b,c,z}_{v,w,n}R^{a,w,y}_{u,k,m}R^{u,v,x}_{i,j,l} = 
\sum_{u,v,w,x,y,z}
R^{a,b,d}_{u,v,x}R^{u,c,e}_{i,w,y}R^{v,w,f}_{j,k,z}R^{x,y,z}_{l,m,n},
\end{align}
where each factor also depends on the type as specified in (\ref{rrrr2}) in general.
For instance the leftmost $R^{d,e,f}_{x,y,z}$ is an element of $R^{DEF}$, whereas 
the next $R^{b,c,z}_{v,w,n}$ is the one for $R^{BCF}$, etc.

There are a couple of tetrahedron equations involving more than one locally non-finite 3D $R$'s on each side, 
which, nevertheless, allow only finitely many quintet $(u,v,w,x,y,z)$ in (\ref{rrrr3}) thanks to the constraint (\ref{ozocg}).
Such types of ABCDEF are the following:
\begin{align}
\text{Type $\quad\;\;$} \quad &   
\text{$\qquad\qquad \qquad $Tetrahedron equation}
\nonumber \\
\text{OOZOOO}:  \quad &R_{456}^{OOO}R_{236}^{OZO}R_{135}^{OZO}R_{124}^{OOO}
=R_{124}^{OOO}R_{135}^{OZO}R_{236}^{OZO}R_{456}^{OOO},
\\
\text{ZOZOOO}:  \quad &R_{456}^{OOO}R_{236}^{OZO}R_{135}^{ZZO}R_{124}^{ZOO}
=R_{124}^{ZOO}R_{135}^{ZZO}R_{236}^{OZO}R_{456}^{OOO},
\\
\text{ZOOOZO}:  \quad &R_{456}^{OZO}R_{236}^{OOO}R_{135}^{ZOZ}R_{124}^{ZOO}
=R_{124}^{ZOO}R_{135}^{ZOZ}R_{236}^{OOO}R_{456}^{OZO},
\\
\text{OZZOOO}:  \quad &R_{456}^{OOO}R_{236}^{ZZO}R_{135}^{OZO}R_{124}^{OZO}
=R_{124}^{OZO}R_{135}^{OZO}R_{236}^{ZZO}R_{456}^{OOO},
\label{nf31}\\
\text{OZOOOZ}:  \quad &R_{456}^{OOZ}R_{236}^{ZOZ}R_{135}^{OOO}R_{124}^{OZO}
=R_{124}^{OZO}R_{135}^{OOO}R_{236}^{ZOZ}R_{456}^{OOZ},
\\
\text{OOZZOO}:  \quad &R_{456}^{ZOO}R_{236}^{OZO}R_{135}^{OZO}R_{124}^{OOZ}
=R_{124}^{OOZ}R_{135}^{OZO}R_{236}^{OZO}R_{456}^{ZOO},
\\
\text{OOZOZO}:  \quad &R_{456}^{OZO}R_{236}^{OZO}R_{135}^{OZZ}R_{124}^{OOO}
=R_{124}^{OOO}R_{135}^{OZZ}R_{236}^{OZO}R_{456}^{OZO},
\label{nf32}\\
\text{OOZOOZ}:  \quad &R_{456}^{OOZ}R_{236}^{OZZ}R_{135}^{OZO}R_{124}^{OOO}
=R_{124}^{OOO}R_{135}^{OZO}R_{236}^{OZZ}R_{456}^{OOZ},
\\
\text{ZOZOZO}:  \quad &R_{456}^{OZO}R_{236}^{OZO}R_{135}^{ZZZ}R_{124}^{ZOO}
=R_{124}^{ZOO}R_{135}^{ZZZ}R_{236}^{OZO}R_{456}^{OZO},
\label{nf33}\\
\text{OZZOOZ}:  \quad &R_{456}^{OOZ}R_{236}^{ZZZ}R_{135}^{OZO}R_{124}^{OZO}
=R_{124}^{OZO}R_{135}^{OZO}R_{236}^{ZZZ}R_{456}^{OOZ}.
\label{zlast}
\end{align}
Note that (\ref{nf31}), (\ref{nf32}), (\ref{nf33}) and (\ref{zlast}) contain {\em three} locally non-finite 3D $R$'s.
To summarize so far,  type OOOOOO and (\ref{z1})--(\ref{zlast}) are the complete list 
of tetrahedron equations of type $\text{ABCDEF} \in \{\text{O,Z}\}^6$ 
which allow finitely many $(u,v,w,x,y,z)$ in (\ref{rrrr3}) enabling us to perform a direct check
for various $(a,b,c,d,e,f)$ and $(i,j,k,l,m,n)$.

When doing so, parameters in the 3D $R$'s are to be chosen with care.
Let us illustrate it along the example (\ref{z4}).
The representations $\pi_O$ and $\pi_Z$ carry the parameter $\mu$ and the quartet $(r,s,t,w)$, respectively.
Thus the tensor components corresponding to 1, 2, 3, 4, 5 and 6 are assigned with the parameters 
$\mu_1, \mu_2, \mu_3, (r_4,s_4,t_4,w_4), \mu_5$ and $\mu_6$, respectively.
$R^{OOO}_{236}$ and $R^{OOO}_{135}$ are given by (\ref{hp2}) by 
replacing $(\mu_1,\mu_2,\mu_3)$ with $(\mu_2,\mu_3,\mu_6)$ and $(\mu_1,\mu_3,\mu_5)$, respectively.
In view of the presence of $R^{ZOO}_{456}$ and Theorem \ref{th:zoo}, we assume
$\frac{\mu_6}{\mu_5} = q^{d}$ for some $d \in \Z$, and take 
$R^{ZOO}_{456}$ to be (\ref{rwoo}) with $(r_1,s_1,t_1,w_1) \rightarrow (r_4,s_4,t_4,w_4)$, 
$\mu_2 \rightarrow \mu_5$ and $\mu_3 \rightarrow  \mu_6$.
Similarly from $R^{OOZ}_{124}$ and Theorem \ref{th:ooz},
we postulate $\frac{\mu_1}{\mu_2} = q^{d'}$ for some $d'\in \Z$,
and $R^{OOZ}_{124}$ is given by (\ref{roow}) with 
$(r_3,s_3,t_3,w_3) \rightarrow (r_4,s_4,t_4,w_4)$ and $d \rightarrow d'$.
With these choices, the tetrahedron equation (\ref{z4}) 
depends on the seven continuous parameters $\mu_1, \mu_3, \mu_5, r_4, s_4, t_4, w_4$ and
the two integer parameters $d, d'$  in addition to the ubiquitous $q$.
Parameters in the other tetrahedron equations are to be tuned similarly.
They can be arbitrary as long as the relevant 3D $R$'s are non-singular, being 
free from the vanishing $q$-shifted factorials in the denominators (if any).

Now we state a conjecture based on computer experiments, indicating
a sort of {\em  coherence}  prevailing the 3D $R$'s obtained in Section \ref{sec:ZO}.
\begin{conjecture}\label{con:r4}
The tetrahedron equations (\ref{z1})--(\ref{zlast}) 
are valid in full generality of parameters.
\end{conjecture}
Typically, equalities have been checked for about 10000 choices of the pairs
$((a,b,c,d,e,f), (i,j,k,l,m,n))$.

\vspace{0.2cm}
Let us turn to the tetrahedron equation in which at least one side of (\ref{rrrr3}) becomes a sum over infinitely many 
$(u,v,w,x,y,z)$'s.
A typical examples is type ZZZZZZ.
A possible regularization in such a circumstance  
is to specialize $q$ to a root of unity and thereby to replace $F$ by a finite dimensional vector space.
For $R^{XXX}$, such a recipe is known \cite{BMS08} to yield the 3D $R$ corresponding to \cite{BB92},
which is closely  related with the generalized Chiral Potts models \cite{BKMS91, DJMM91, BB92, SMS96}.
It is an interesting problem to explore a similar connection for the other 3D $R$'s in this paper.
Remark \ref{re:cro} is a  key to such studies concerning $R^{OOO}$.

Finally the whole setting concerning the quantized Yang-Baxter equation 
$RLLL=LLLR$ in this paper has a natural analogue 
in the {\em quantized reflection equation}  $K(GLGL)=(LGLG)K$ \cite{KP18} which is 
related to the quantized coordinate rings of type $B$ and $C$.
It awaits a discovery of new 3D $K$'s  different from the known ones in 
\cite{KO12} and \cite[Chap. 5 \& 6]{Ku22}.

\appendix

\section{Explicit form of $R_{456}L^{Z}_{236}L^{Z}_{135}L^{Z}_{124}
=L^{Z}_{124}L^{Z}_{135}L^{Z}_{236}R_{456}$}\label{app:eqs}

We write down (\ref{ZZZ}) explicitly together with 
the corresponding choice of $(abcijk)$ in (\ref{qybe}) or in Figure \ref{fig:qybe}.
As mentioned around (\ref{rdef}), there are 18 non-trivial cases.
To save the space, we write 
$Y_\alpha = Z^{-1}(r_\alpha s_\alpha-t_\alpha^2w_\alpha X^2)$.
\begin{align}
(001001): \,
&R(1\otimes  X\otimes  X) = (1\otimes  X \otimes  X)R,
\\
(001010):\,&R\bigl(
r_2t_1X \otimes  1 \otimes  Y_3
+t_3Z\otimes  Y_2 \otimes  X\bigr) = r_1t_2\bigl( 1 \otimes  X \otimes  Y_3\bigr)R,
\\
(001100):\, &R\bigl(
-qt_1t_3w_1X \otimes Y_2\otimes X
+r_2 Y_1 \otimes 1 \otimes Y_3\bigr)
= r_1r_3\bigl(1 \otimes Y_2 \otimes 1\bigr)R,
\\
(010001):\, &
r_1t_2R(1 \otimes X \otimes Z) = 
\bigl(r_2t_1X \otimes 1 \otimes Z 
+ t_3Y_1 \otimes Z\otimes X \bigr)R,
\\
(010010):\, &R\bigl(
qr_2t_1t_3w_3 X \otimes 1 \otimes X - Z \otimes Y_2 \otimes Z\bigr)
= \bigl(
qr_2t_1t_3w_3 X \otimes 1 \otimes X
-Y_1 \otimes Z \otimes Y_3
\bigr)R,
\\
(010100):\, &R\bigl(
t_1w_1 X \otimes Y_2 \otimes Z
+ r_2t_3w_3Y_1\otimes 1 \otimes X\bigr)
= r_3t_2w_2 \bigl(Y_1 \otimes X \otimes 1\bigr) R,
\\
(011011):  \,
&R(X \otimes  X \otimes  1) = (X \otimes  X \otimes  1)R,
\\
(011101):\, &s_3t_2  R\bigl(
Y_1\otimes X \otimes 1\bigr)
= \bigl(
t_1X \otimes Y_2 \otimes Z 
+ s_2t_3 Y_1\otimes 1 \otimes X
\bigr)R,
\\
(011110):\, 
&s_1s_3 R\bigl(
1 \otimes Y_2  \otimes 1\bigr)
= \bigl(-
q t_1t_3w_3 X \otimes Y_2 \otimes X
+s_2Y_1\otimes 1 \otimes Y_3
\bigr)R,
\\
(100001): \,
&r_1r_3 R(1\otimes  Z \otimes  1) =  (-q t_1t_3w_1 X \otimes  Z \otimes  X +r_2 Z \otimes  1 \otimes  Z)R,
\\
(100010):\,
&r_3t_2w_2R(Z \otimes X \otimes 1)
=\bigl( t_1w_1 X \otimes Z \otimes Y_3
+ r_2t_3w_3 Z \otimes 1 \otimes X\bigr)R,
\\
(100100):  \,
&R(X \otimes  X \otimes  1) = (X \otimes  X \otimes  1)R,
\\
(101011): 
& R (t_1X \otimes  Z \otimes  Y_3 + s_2 t_3 Z \otimes  1 \otimes  X) 
= s_3 t_2 (Z \otimes  X \otimes  1)R, 
\\
(101101):\, &R\bigl(
-qs_2t_1t_3w_1 X \otimes 1 \otimes X 
+ Y_1 \otimes Z \otimes Y_3\bigr)
=\bigl( -q s_2 t_1t_3w_1X \otimes 1 \otimes X
+ Z \otimes Y_2 \otimes Z\bigr)R,
\\
(101110):\, &R\bigl(s_1t_2w_2 1 \otimes X \otimes Y_3)
= 
\bigl(s_2t_1w_1X \otimes 1 \otimes  Y_3
+ t_3w_3 Z \otimes Y_2 \otimes X\bigr)R,
\\
(110011): \,
&R(-qt_1t_3w_3 X \otimes  Z \otimes  X + s_2 Z \otimes  1 \otimes  Z) = s_1s_3 (1 \otimes  Z \otimes  1)R,  
\\
(110101):\,
&R\bigl (t_3 w_3 Y_1 \otimes  Z \otimes  X + s_2 t_1 w_1 X \otimes  1 \otimes  Z\bigr) 
= s_1 t_2 w_2 (1 \otimes  X \otimes  Z)R,
\\
(110110): \,
&R(1\otimes  X\otimes  X) = (1\otimes  X \otimes  X)R.
\end{align}

\section*{Acknowledgments}
A.Y. is supported by Grants-in-Aid for Scientific Research No.~21J11742 from JSPS.

\end{document}